\documentclass[runningheads]{llncs}

\usepackage{geometry}
\usepackage[T1]{fontenc}
\usepackage{amsmath}
\usepackage{amssymb}
\usepackage{enumitem}
\usepackage{hyperref}
\usepackage{cleveref}
\usepackage[ruled]{algorithm2e} 

\usepackage{graphicx}
\usepackage{xcolor}
\usepackage{url}
\usepackage{diagbox}
\usepackage[comma, numbers]{natbib}

\newenvironment{proofsketch}{\trivlist\item\emph{Sketch of Proof.}}{\endtrivlist}
\newenvironment{informaltheorem}{\trivlist\item\textbf{Theorem (Informal).}}{\endtrivlist}

\usepackage{bm}

\newcommand\E{\mathop{\mathbb E}}
\newcommand\I{\mathop{\mathbb I}}
\DeclareMathOperator{\secmax}{{\mathrm{secmax}}}

\newcommand\bbR{{\mathbb R}}
\newcommand\calF{{\mathcal F}}
\newcommand\calA{{\mathcal A}}

\newcommand{\Myer}{\mathrm{Myer}}
\newcommand{\WEL}{\mathrm{WEL}}
\newcommand{\SPA}{\mathrm{SPA}}
\newcommand{\VCG}{\mathrm{VCG}}
\newcommand{\REV}{\mathrm{REV}}

\usepackage{xcolor}

\title{Ex-Ante Truthful Distribution-Reporting Mechanisms}
\author{Xiaotie Deng \inst{1}
\and Yanru Guan \inst{1}
\and Ningyuan Li \inst{1}
\and Zihe Wang \inst{2}
\and Jie Zhang \inst{3}}
\institute{Peking University, China \\
\email{xiaotie@pku.edu.cn, piscesguan@stu.pku.edu.cn, liningyuan@pku.edu.cn}\and
Renmin University of China, China\\
\email{wang.zihe@ruc.edu.cn}\and
University of Bath, United Kingdom\\
\email{jz2558@bath.ac.uk}
}

\begin{document}

\maketitle

\begin{abstract}
This paper studies mechanism design for revenue maximization in a distribution-reporting setting, where the auctioneer does not know the buyers' true value distributions. Instead, each buyer reports and commits to a bid distribution in the ex-ante stage, which the auctioneer uses as input to the mechanism. Buyers strategically decide the reported distributions to maximize ex-ante utility, potentially deviating from their value distributions. As shown in previous work, classical prior-dependent mechanisms such as the Myerson auction fail to elicit truthful value distributions at the ex-ante stage, despite satisfying Bayesian incentive compatibility at the interim stage. We study the design of ex-ante incentive compatible mechanisms, and aim to maximize revenue in a prior-independent approximation framework. 
We introduce a family of \emph{threshold-augmented mechanisms}, which ensures ex-ante incentive compatibility while boosting revenue through ex-ante thresholds. Based on these mechanisms, we construct the \emph{Peer-Max Mechanism}, which achieves an either-or approximation guarantee for general non-identical  distributions. Specifically, for any value distributions, its expected revenue either achieves a constant fraction of the optimal social welfare, or surpasses the second-price revenue by a constant fraction, where the constants depend on the number of buyers and a tunable parameter. We also provide an upper bound on the revenue achievable by any ex-ante incentive compatible mechanism, matching our lower bound up to a constant factor. Finally, we extend our approach to a setting where multiple units of identical items are sold to buyers with multi-unit demands.

\end{abstract}

\section{Introduction}
Auction design has long been a cornerstone of economic theory and mechanism design, with applications ranging from government auctions for spectrum licenses to online marketplaces and computational advertising. At its core, auction design aims to allocate resources efficiently or to maximize revenue for the auctioneer, typically by incorporating information about buyers' preferences. Foundational work such as Myerson’s seminal auction theory \cite{bc0ee6e0-c3d6-3130-9f09-cf95a8a3f3a4} provides a rigorous framework for revenue-optimal auctions, under the assumption that the auctioneer has precise knowledge of the buyers’ value distributions.

In practice, however, this assumption is often unrealistic. Buyer valuations are often private and shaped by unobservable factors, making it difficult or impossible for the auctioneer to know the underlying value distributions. Consequently, classical auction theory cannot be directly applied, posing a significant challenge to revenue-maximizing mechanism design. This has motivated the development of prior-independent mechanisms \cite{pimd-Dhangwatnotai15}, which seek strong performance guarantees without requiring distributional knowledge. However, the absence of distributional information hinders revenue maximization even in simple settings. Prior-independent mechanisms often require restrictive assumptions, such as identical or regular value distributions, to deliver satisfactory revenue guarantees.

In many real-world settings, particularly in computational advertising, auctioneers such as ad platforms interact repeatedly with the same set of buyers across numerous rounds. To optimize mechanism performance, platforms commonly use historical bid data to estimate buyers’ value distributions. However, this creates incentives for buyers to strategically distort their bids in order to influence future mechanism outcomes, especially when they understand that current bids affect future distribution estimates.
This strategic behavior is formally captured in prior research \cite{tang2018price,deng2020www,ChenZH2024www} through a distribution-reporting model, where buyers report and commit to fake distributions in the ex-ante stage. Analyses demonstrate that such strategic manipulation degrades the revenue performance of classical prior-dependent auctions (e.g., Myerson auction) to the level of prior-independent auctions (e.g., first-price auction).

In this paper, we study the mechanism design problem under the distribution-reporting model.
The central question is whether mechanisms can be designed to truthfully elicit buyers’ private value distributions at the ex-ante stage, while using this reported distribution information to secure strong revenue guarantees. We answer this question affirmatively. 
In our model, while a buyer's true value distribution is private, it is assumed to lie within a known distribution class. Each buyer strategically selects and reports a distribution from this class at the ex-ante stage, committing to bid according to this distribution, which the auctioneer then uses to determine allocation and pricing rules. 
We design mechanisms that are ex-ante incentive compatible, ensuring that buyers maximize their ex-ante expected utility by truthfully reporting their value distributions. To evaluate the revenue performance, we aim for worst-case revenue guarantees akin to prior-independent mechanism design. Since canonical benchmarks such as Myerson revenue may be unattainable, we evaluate our mechanisms against a well-defined benchmark tailored to this environment, which reflects what is realistically achievable while still capturing the goal of strong revenue performance.

\subsection{Paper Overview}

\subsubsection{Most Relevant Related Work.} \citet{tang2018price} introduce the distribution-reporting game modeling the buyers' strategic behavior in a setting where the auctioneer infers value distributions from reported bids. Each buyer has a true type distribution, from which their private type is drawn, but can strategically commit to a fake distribution. Buyers then behave rationally as if their values were drawn from the fake distribution, causing the mechanism to adopt the fake distribution as prior. They analyze the equilibrium in the resulting distribution-reporting game, and show that classical prior-dependent mechanisms, such as Myerson's auction and the second-price auction with a monopoly reserve, fail to be incentive compatible for buyers' ex-ante utilities, when buyers can manipulate the distribution the mechanism relies on. This loss of incentive compatibility significantly degrades the performance of these mechanisms, making them no better than simple prior-independent mechanisms like first-price or second-price auctions. Our mechanism design model is adapted from their distribution-reporting model, replacing the Bayesian incentive compatibility\footnote{Note that Bayesian incentive compatibility is defined on buyers' interim utility, and does not imply ex-ante IC.} assumption with ex-ante incentive compatibility requirement. While \citet{tang2018price} focus on analyzing the impact of fake distributions on classical mechanisms, our work instead investigates how to design mechanisms that are ex-ante incentive compatible within the distribution-reporting model.

In the context of prior-independent mechanism design, \citet{GuruganeshMWW24} recently study the problem where buyers have non-identical but regular value distributions. For revenue approximation, they adopt an either-or framework: designing a mechanism that either achieves a constant-factor approximation to the expected revenue of Myerson’s auction or exceeds the second-price auction revenue by a multiplicative factor. This approach is motivated by the observation that in heterogeneous settings, no prior-independent mechanism can consistently approximate Myerson’s revenue, nor can it always outperform the second-price auction.
While their work focuses on classical incentive compatibility constraints in single-shot auctions, we study a different setting in which buyers report value distributions and the mechanism is required to be ex-ante incentive compatible. We also adopt an either-or approximation framework, but use a stronger variant of their benchmarks: while both works include the second-price auction revenue ($\SPA$) as one component, we replace Myerson’s revenue with the optimal social welfare ($\WEL$). Notably, $\WEL$ always upper bounds Myerson’s revenue and the gap between the two can be arbitrarily large. Technically, both papers employ randomized threshold techniques in mechanism construction, but the analysis in \citet{GuruganeshMWW24} does not carry over to our setting, as it critically relies on the regularity of distributions and focuses on benchmarks different from ours (i.e., $\Myer$ rather than $\WEL$).

\subsubsection{Summary of Main Results.}

We use the \emph{bid-reporting} setting to refer to  environments in which each buyer submits a single value as a bid to the mechanism, as in classical auction literature. In contrast, the \emph{distribution-reporting} setting, which is the focus of this paper, refers to environments in which each buyer reports a value distribution that the mechanism uses as input to determine outcomes.

To design ex-ante incentive compatible mechanisms in the distribution-reporting setting, we introduce the class of \emph{threshold-augmented mechanisms}. These mechanisms are built on top of an incentive compatible prior-independent mechanism from the bid-reporting setting and enhance revenue by applying a threshold to each buyer’s ex-ante utility: buyers whose utility exceeds the threshold retain the original allocation and their payment is added the threshold value, while those below the threshold are excluded. 
An immediate application of threshold-augmented mechanisms is in achieving full surplus extraction under the i.i.d. setting. When all buyers share the same (but unknown) value distribution, we construct a simple threshold-augmented mechanism that guarantees the seller's revenue equals the optimal social welfare.

Our main result is the design of ex-ante incentive compatible distribution-reporting mechanisms that achieve an either-or approximation guarantee with respect to the optimal welfare ($\WEL$) and the second-price auction revenue ($\SPA$) in the setting where buyers have non-identical and general value distributions.

\begin{informaltheorem}
    For $n\geq 2$ buyers and any $K\geq 1$, there exists an ex-ante IC distribution-reporting mechanism that guarantees an expected revenue of 
$$\min \Bigl\{ \frac1{24(K+\log_2 n)}\cdot \WEL,2^K\cdot \SPA \Bigr\},$$ for general value distributions.
\end{informaltheorem}

To establish the tightness of the lower bound, we provide a matching upper bound, up to constant factors. 
\begin{informaltheorem}
    For $n\geq 2$ buyers and any $K\geq 6$, for general value distributions, no ex-ante IC distribution-reporting mechanism can always guarantee a revenue of 
$$\min \Bigl\{ \frac{64}{K+\log_2 n}\cdot \WEL,2^K\cdot \SPA \Bigr\}.$$
\end{informaltheorem}
Note that this upper bound can be extended to the regular distribution setting with constant factor losses.

Finally, we extend our main positive result to the setting where $m$ units of identical items are sold to $n$ multi-demand buyers with unknown, non-identical value distributions, using a distribution-reporting mechanism. Note that this result is also tight up to constants, since the above upper bound directly carries over to this setting as a special case.
\begin{informaltheorem}
    For the multi-unit multi-demand setting, for any $n\geq 2$ and $K\geq 1$, there exists an ex-ante IC distribution-reporting mechanism that guarantees an expected revenue of 
$$\min \Bigl\{ \frac1{24(K+\log_2 n)}\cdot \WEL,2^K\cdot \VCG \Bigr\},$$ for general value distributions.
\end{informaltheorem}

Most proofs are deferred to the Appendix. For key results, we include proof sketches in the main text when space allows.

\subsubsection{Outline of Technical Approach.}

To ensure ex-ante incentive compatibility while exploiting the reported distribution information, we propose a family of distribution-reporting mechanisms named threshold-augmented mechanisms. A threshold-augmented mechanism builds on a prior-independent, incentive compatible mechanism from the bid-reporting setting by applying a threshold on each buyer's ex-ante utility. This threshold is computed as a function of the value distributions reported by the other buyers. It plays the role of an ex-ante entry fee: a buyer whose expected utility under the original bid-reporting mechanism exceeds her threshold proceeds with the original allocation and payment rules but is additionally charged an amount equal to the threshold. A buyer whose utility falls below the threshold is excluded from the auction, receiving no allocation and making no payment. However, her bid is still simulated in the original mechanism when computing the allocation and payment outcomes for the other buyers, so that the overall competitive landscape is unaffected by the exclusion of buyers. This structure preserves incentive compatibility because the underlying mechanism is incentive compatible and the threshold imposed on each buyer is independent of her own report.

To achieve the worst-case revenue guarantees, we employ a properly designed randomization over threshold-augmented mechanisms. 
The use of randomized thresholds, along with our either-or approximation benchmark, is inspired by techniques developed in prior-independent mechanism design, particularly the work of \citet{GuruganeshMWW24}. However, their analysis is tailored to settings where buyers have identical or regular value distributions. In contrast, our results hold under general, non-identical distributions. Moreover, the ex-ante nature of our distribution-reporting setting necessitates different approximation target and requires new techniques distinct from those used in the interim, bid-reporting context.

Specifically, we design the \emph{Peer-Max mechanism}, which sets the threshold for each buyer based on the expected maximum value among the other buyers, scaled by a randomly selected multiplicative factor from a geometrically separated set of size $O(K + \log n)$. The analysis considers two distinct cases: either a single buyer contributes a large fraction of the social welfare under the second-price auction, or the welfare is more evenly distributed and each buyer contributes only a small fraction. In the first case, we directly lower bound the revenue obtained from this dominant buyer. In the second, we aggregate revenue contributions from all buyers whose welfare contributions are not too small. To maximize revenue in both cases, the Peer-Max mechanism selects thresholds that approximate each buyer’s ex-ante utility from below, increasing the likelihood that they remain in the auction and contribute to the revenue. This can be viewed as ``guessing'' each buyer’s ex-ante utility and setting the entry threshold just low enough to retain them. The $\log n$ factor in the revenue bound arises from accommodating buyers whose welfare contribution fractions vary from $\Theta(1/n)$ to $\Theta(1)$.


To prove a matching upper bound on the either-or approximation, we construct a distribution over instances where each buyer's value distribution is randomly drawn from a family of distributions. Each distribution in this family is a scaled version of a simple equal-revenue type, with the scaling factor varying geometrically. The probabilities assigned to these distributions are chosen so that the induced distribution over scaling factors approximates a discrete equal-revenue distribution. In this setting, reporting a distribution effectively amounts to reporting the scaling factor. This reduction implies that any ex-ante incentive compatible distribution-reporting mechanism is equivalent to a mechanism that is incentive compatible in the bid-reporting setting. We then use this equivalence to upper bound the expected revenue of any such mechanism under our constructed distribution over instances.


\subsection{Further Related Literature}
The field of auction design was fundamentally shaped by the seminal work of \citet{bc0ee6e0-c3d6-3130-9f09-cf95a8a3f3a4}, which laid the groundwork for formal approaches to maximizing seller revenue. More recently, a significant line of research has focused on designing mechanisms that do not rely on prior knowledge of buyers’ value distributions. This direction emerged in part because requiring prior information is often unrealistic in practice and poses technical challenges in analysis. We summarize several strands of work that have developed within this broader effort.

One strand of work studies the incentive for bidders to strategically manipulate their bid data to improve long-term utility under classical prior-dependent mechanisms, and characterize the impact of strategic manipulation on the mechanism outcomes. 
The distribution-reporting model introduced by \citet{tang2018price} has been adapted to investigate the strategic  manipulation of private distribution in various settings. 
\citet{deng2020www}
adapts this model to study private data manipulation in sponsored search auctions, showing how strategic buyers can distort reported distributions to influence outcomes, complicating revenue optimization.
\citet{ChenZH2024www} extends this model to a budget-constrained auction setting, and establish strategic equivalence and revenue dominance results between the Bayesian revenue-optimal mechanism and budget-constrained variations of first-price and second-price auctions.
\citet{deng2020neurips} 
study a different model where bidders strategically manipulate a limited number of samples to exploit the ERM algorithm, showing upper bounds on the impact of strategic manipulation.
To the best of our knowledge, our work is the first to formally study the design of an ex-ante IC mechanism within the distribution-reporting model.

Another strand of work is on prior-independent auction design, exploring worst-case revenue maximization facing unknown value distributions. \citet{pimd-Dhangwatnotai15} introduce the goal to design a prior-independent auction whose expected revenue under every distribution approximates the optimal auction for that distribution, and show that the second-price auction achieves at least $(n-1)/n\geq\frac12$ approximation of the revenue of Myerson auction. 
\citet{pimd-DBLP:conf/sigecom/FuILS15} propose a randomized prior-independent auction that achieves strictly better approximation guarantees. 
\citet{pimd-Allouah20} improve the upper and lower bounds on the approximation ratios, and show that second-price auction achieves optimal approximation ratio for MHR distributions.
\citet{pimd-DBLP:conf/focs/HartlineJL20} close the gap for two bidders with i.i.d. regular distributions, showing a tight approximation factor of 0.524.
Diverging from the i.i.d. setting, \citet{GuruganeshMWW24} study the setting with non-identical distributions, introducing the either-or approximation benchmark to bypass the inapproximability of canonical benchmarks like Myerson revenue.

Another line of work focuses on optimizing the auction rule through a dynamic learning process and its connection to no-regret learning. This research explores how to adapt reserve prices over time using observed bids and how to maintain long-term revenue despite strategic bidder behavior.
\citet{Amin13neurips} 
study learning a buyer’s value distribution in repeated posted-price settings, modeling the buyer as forward-looking and maximizing long-term surplus. They design seller strategies with no-regret guarantees under discounted utility, showing that lower discounting leads to higher regret, up to linear regret when no discounting is assumed.
\citet{kanoria14wine}
investigate dynamic reserve pricing in repeated second-price auctions and show that adapting reserves using bid history improves revenue under unknown distributions, proposing an approximately incentive-compatible and asymptotically optimal mechanism.
\citet{Abernethy19neurips}
design a mechanism that incorporates differential privacy into the learning process, ensuring approximate revenue optimality while maintaining strategic robustness and discouraging manipulation.
\citet{Drutsa20icml}
study repeated second-price auctions with reserve prices and propose a learning algorithm with $O(\log\log T)$ strategic regret against manipulation.

\section{Model and Preliminary}\label{sec:model}
\subsection{Distribution-reporting mechanism design}
There are $n$ buyers, each having a value distribution $F_i$. Each buyer's distribution is private, but lies in a publicly known distribution class $\calF_i$. Let $\bm{\calF}:=\calF_1\times\cdots\times\calF_n$ denote the Cartesian product of these distribution classes, and we call $\bm{F}=(F_1,\cdots,F_n)\in\bm{\calF}$ a distribution profile. Specifically, we also refer to a distribution profile of true value distributions as an \textit{instance}.

We denote the set of non-negative real numbers as $\bbR_+=[0,+\infty)$. Let $\mathcal{A}$ denote the class of all distributions on $\bbR_+$, and $\mathcal{A}^R$ denote the class of all regular distributions on $\bbR_+$. Formally, a distribution is regular if its virtual value function $\varphi(x)=\frac{1-F(x)}{f(x)}$ is non-decreasing, where $F(x)$ and $f(x)$ are its cumulative distribution function and probability density function, respectively.

We study the mechanism design problem under the distribution-reporting model adapted from \cite{tang2018price}, where each buyer submits a distribution at the ex-ante stage, and commits to bid according to this distribution. 

In the distribution-reporting model, the strategic action of each buyer is submitting a bid distribution $B_i\in\calF_i$, and privately deciding a (possibly randomized) mapping $\beta_i$ as the bidding function, which maps $F_i$ to $B_i$. That is, when $v_i$ follows $F_i$, the distribution of $\beta_i(v_i)$ is $B_i$.
We call $\bm{B}=(B_1,\cdots,B_n)\in\bm{\calF}$ a bid distribution profile. The decision of $\beta_i$ is made ex-ante, i.e., before knowing the realization of the value $v_i$.


Throughout the paper, we use a \textit{bid-reporting mechanism} to refer to a prior-independent mechanism, 
which consists of an allocation rule and a payment rule mapping a bid profile to the allocation and payment outcomes.
In contrast, a \textit{distribution-reporting mechanism} takes the bid distributions as input in ex-ante stage to decide the allocation and payment rules.

\begin{definition}[Distribution-reporting mechanism]
A distribution-reporting mechanism $M=(x,p)$ consists of the allocation rule $x_i:\bbR_+^n\times \bm{\calF}\to[0,1]$ and the payment rule $p_i:\bbR_+^n\times \bm{\calF}\to\bbR$ for all $i\in[n]$.

Specifically, given the submitted bid distribution profile $\bm{B}$, the mechanism $M$ executes a bid-reporting mechanism which maps a bid profile $\bm{b}=(b_1,\cdots,b_n)\in\bbR_+^n$ to the allocation and payment outcomes, i.e., allocating the item to buyer $i$ with probability $x_i(\bm{b};\bm{B})$, while charging payment $p_i(\bm{b};\bm{B})$.

We say $M$ is feasible (in the single-item setting) if it holds that
$\sum_{i\in[n]}x_i(\bm{b};\bm B)\leq 1$ for any $\bm{b}\in\bbR_+$ and $\bm{B}\in \bm{\calF}$. 
Note that we will generalize the definition of feasibility in Section \ref{sec:generalization}.

We say $M$ is allocation-monotone if $x_i(b_1,\cdots,b_n;\bm B)$ is always weakly increasing in $b_i$.
\end{definition}

The timeline of the distribution-reporting model is as follows:
\begin{enumerate}
    \item \textbf{Ex-ante Stage}
    \begin{enumerate}
        \item The distribution classes $\bm{\calF}=(\calF_1,\cdots,\calF_n)$ are publicly known, while each buyer's true value distribution $F_i$ is private, and the values $v_i$ is unknown.
        \item The seller decides and announces a distribution-reporting mechanism $M=(x,p)$.
        \item Each buyer $i\in[n]$ decides and submits a bid distribution $B_i\in\calF_i$ to the mechanism, while privately deciding and committing to a mapping $\beta_i$ that maps $F_i$ to $B_i$.
    \end{enumerate}
    \item \textbf{Interim Stage}
    \begin{enumerate}
        \item For each $i\in[n]$, $v_i$ is drawn from $F_i$ and become privately known to buyer $i$.
        \item Each buyer $i$ submits $b_i=\beta_i(v_i)$ to the mechanism.
    \end{enumerate}
    \item \textbf{Ex-post Stage}
    \begin{enumerate}
        \item The mechanism allocates the item such that each buyer $i\in[n]$ is allocated with probability $x_i(\bm{b};\bm{B})$.
        \item Each buyer $i\in[n]$ is charged payment $p_i(\bm{b};\bm{B})$ by the mechanism.
    \end{enumerate}
\end{enumerate}

For convenience, we sometimes represent a distribution by its corresponding quantile function. For a random variable $v$ following some distribution $F$, its quantile is defined as $q=1-F(v)$, where $F(v)$ is the CDF. Viewing $v$ as a function of $q$, the distribution $F$ can be uniquely represented by the quantile function $v(\cdot)$ defined as $v(q)=F^{-1}(1-q)$. Notably, for any distribution $F$, when $q\sim U[0,1]$, the distribution of $v(q)$ is exactly $F$. In particular, we often denote the quantile functions of $F_i$ and $B_i$ by $v_i(\cdot)$ and $b_i(\cdot)$, respectively.

Using the notation of quantiles, when the bid distribution $B_i$ is given, the mapping $\beta_i$ can be equivalently represented by an arrangement function $\phi_i:[0,1] \rightarrow [0,1]$ that maps each quantile in $F_i$ to a quantile in $B_i$. That is, we can assume each buyer is assigned a quantile $q_i$ drawn from $U[0,1]$ in the interm stage, so that her value is $v_i(q_i)$, and she submits bid $b_i(\phi_i(q_i))$, where $v_i(\cdot)$ and $b_i(\cdot)$ denote the quantile functions of $F_i$ and $B_i$. It is necessary that $\phi_i$ preserves $U[0,1]$, i.e. the distribution of $\phi_i(q_i)$ is $U[0,1]$ when $q_i\sim U[0,1]$. 

Given a distribution-reporting mechanism $M$ and an instance of true value distributions $\bm{F}$, each buyer strategically decide her bid distribution $B_i$ and arrangement function $\phi_i$ (representing $\beta_i$) so as to maximize her ex-ante utility.
\begin{definition}[Ex-ante utility]\label{def:ex-ante utility}
Given a distribution-reporting mechanism $M$, a bid distribution profile $\bm B=(B_1,\cdots,B_n)$, the ex-ante utility of a buyer $i$ with true distribution $F_i$ and arrangement function $\phi_i$ is
\begin{align*}
   U^M_i(\bm B,\phi_i;F_i)=   &\E_{q_i\sim U[0,1],\bm{b}_{-i}\sim \bm{B}_{-i}}[x_i(b_i(\phi_i(q_i)),\bm{b}_{-i};\bm B)\cdot v_i(q_i)-p_i(b_i(\phi_i(q_i)),\bm{b}_{-i};\bm B)]\\
   =&\E_{q_i\sim U[0,1],\bm{b}_{-i}\sim \bm{B}_{-i}}[x_i(b_i(\phi_i(q_i)),\bm{b}_{-i};\bm B)\cdot v_i(q_i)]-\E_{\bm{b}\sim \bm{B}}[p_i(\bm{b};\bm B)],
\end{align*}
where $v_i(\cdot)$ and $b_i(\cdot)$  are the quantile function of $F_i$ and $B_i$, respectively.
\end{definition}

Notably, as the buyers aim to maximize the ex-ante utility, they only care about the expected payment $\E_{\bm{b}\sim \bm{B}}[p_i(\bm{b};\bm B)]$, which only depends on the bid distribution profile. Therefore, without loss of generality, we assume the payment rule is independent of the realized bids, and simplify the payment rule to $p_i(\bm{B})$ throughout the paper.


Without loss of generality, we always assume the distribution-reporting mechanism is allocation-monotone. Under this assumption, by \Cref{lemma:lemma3.1}, the dominant strategy of $\phi_i$ is always the weakly monotone function mapping $F_i$ to $B_i$, or equivalently, the identity mapping in the quantile space.
\begin{lemma}[Lemma 3.1 in \cite{tang2018price}]\label{lemma:lemma3.1}
    For any allocation-monotone distribution-reporting mechanism $M$ and fixed bid distribution profile $\bm B$, for any buyer $i\in[n]$ with value distribution $F_i$, the ex-ante utility $U^M_i(\bm B,\phi_i;F_i)$ is maximized when the arrangement function is the identity mapping $\phi_i(q)=q, ~\forall q\in [0,1]$.
\end{lemma}
By \Cref{lemma:lemma3.1}, we always assume that $\phi_i$ is the identity mapping, and simplify the buyer's strategy to solely deciding a bid distribution $B_i$.
In particular, the ex-ante utility is simplified to
\begin{align*}
   U^M_i(\bm B;F_i)=\E_{q_i\sim U[0,1],\bm{b}_{-i}\sim \bm{B}_{-i}}[x_i(b_i(q_i),\bm{b}_{-i};\bm B)\cdot v_i(q_i)]-p_i(\bm B).
\end{align*}

\begin{definition}
A distribution-reporting mechanism $M$ is ex-ante incentive compatible (ex-ante IC), if for any buyer $i\in[n]$, for any true distribution $F_i\in\calF_i$, for any $\mathbf{B}_{-i}$, the ex-ante utility $U^M_i(B_i,\bm{B}_{-i};F_i)$ is always maximized by truthfully reporting $B_i=F_i$, i.e., 
$U^M_i(F_i,\bm{B}_{-i};F_i)=\max_{B_i\in\calF_i}U^M_i(B_i,\bm{B}_{-i};F_i)$.

A distribution-reporting mechanism $M$ is ex-ante individual rational (ex-ante IR), if for any true distribution profile $\bm{F}\in\bm{\calF}$, for any buyer $i\in[n]$, the ex-ante utility obtained by truthful report is non-negative, i.e., $U^M_i(\bm{F};F_i)\geq 0$.
\end{definition}


We say a distribution-reporting mechanism $M$ is \textit{valid}, if it is ex-ante IC, ex-ante IR, and feasible.
When a distribution-reporting mechanism $M$ is valid, we assume all buyers truthfully report $B_i=F_i$, and we define its expected revenue as follows.
\begin{definition}
    For any valid distribution-reporting mechanism $M=(x,p)$, for any instance $\bm{F}\in\bm{\calF}$, we define the expected revenue of $M$ as
\begin{align*}
\mathrm{REV}(M,\bm{F})&=\sum_{i\in[n]} p_i(\bm{F}).
\end{align*}
\end{definition}

\subsection{Performance measurement}
In the distribution-reporting model, our goal is to design mechanisms that maximize revenue without knowing the buyers’ true value distributions. Because no single mechanism can guarantee optimal revenue across all possible distributions, we adopt a worst-case approximation approach from the prior-independent mechanism design framework. This allows us to evaluate a mechanism’s performance by comparing its revenue to  suitable benchmarks.

We adopt the either-or approximation framework inspired by \citet{GuruganeshMWW24}, though with a key difference in the choice of benchmarks: while we retain the second-price auction revenue $\SPA$ as one benchmark, we adopt the optimal social welfare $\WEL$ as the other benchmark. For any instance $\bm{F} \in \bm{\calF}$, these two benchmarks are defined as:
\begin{align*}
    \WEL(\bm{F})&=\E_{v_i\sim F_i,\forall i}[\max_{i\in[n]} v_i],\\
    \SPA(\bm{F})&=\E_{v_i\sim F_i,\forall i}[\secmax_{i\in[n]} v_i],
\end{align*}
where $\secmax_{i \in [n]} v_i$ denotes the second-highest value among the sampled value profile.

We now formally define the notion of an either-or approximation:
\begin{definition}
    A distribution-reporting mechanism $M$ is said to achieve an either-or approximation with respect to $\WEL$ and $\SPA$ with parameters $\delta_1 > 0$ and $\delta_2 > 0$ if, for every instance $\bm{F} \in \bm{\calF}$, it holds that
    $$
    \REV(M; \bm{F}) \geq \min\{\delta_1 \cdot \WEL(\bm{F}),\ (1 + \delta_2) \cdot \SPA(\bm{F})\}.
    $$
    In other words, for any instance, the expected revenue of $M$ either approximates the optimal social welfare up to a constant factor, or strictly outperforms the second-price auction revenue by a constant margin.
\end{definition}

In comparison, \citet{GuruganeshMWW24} study prior-independent mechanisms that achieve an either-or approximation with respect to Myerson revenue and second-price revenue under regularity assumptions, and establish matching upper bounds. Specifically, for any parameter $\tau>e$, they prove that there exists a prior-independent mechanism that guarantees $\min\{\Omega(\frac{1}{\ln\tau}) \cdot \Myer,\tau\cdot \SPA\}$ revenue for all regular distributions, where $\Myer$ denote the expected revenue of Myerson auction. They also show a matching upper bound up to constant factors.

It is well-known that $\WEL(\bm{F})$ always upper bounds $\Myer(\bm{F})$. Therefore, our choice of benchmark is generally more demanding than that in \cite{GuruganeshMWW24}.

Our use of this either-or benchmark is motivated by two key considerations, which we outline below.
\begin{enumerate}
\item \textbf{No ex-ante IC distribution-reporting mechanism can guarantee an $\epsilon$-approximation to $\WEL$, nor can it guarantee a $(1+\epsilon)$-approximation to second-price revenue, for any $\epsilon>0$.}  
We prove these limitations in \Cref{lemma:impossible-WELandSPA}. The core obstacle arises from the ex-ante incentive compatibility constraint, which forces any distribution-reporting mechanism to behave like a bid-reporting IC mechanism when facing degenerate distributions.

\item \textbf{No prior-independent mechanism can achieve an either-or approximation with respect to $\WEL$ and $\SPA$ using parameters $\delta_1 = \delta_2 = \epsilon$ for any $\epsilon > 0$, even when $n=2$, and all distributions are identical and regular.}
We establish this result in \Cref{lemma:impossible-prior-independent}. This demonstrates a separation between prior-independent and distribution-reporting mechanisms in terms of revenue performance. In contrast to this impossibility, our main result shows that distribution-reporting mechanisms can achieve either-or approximation with respect to $\WEL$ and $\SPA$ using constant parameters, for any fixed $n \geq 2$, without requiring regularity. The key advantage lies in their ability to utilize the reported distributions to extract additional revenue.

\end{enumerate}

We remark that existing literature  has extensively studied the i.i.d. setting of prior-independent bid-reporting mechanism design. However, designing a distribution-reporting mechanism in the identical distribution setting is simple, as the distribution of a buyer will be revealed by other buyers' reports, allowing the mechanism to extract the entire social welfare as revenue (see \Cref{thm:iid-REV=WEL}).
Therefore, our main results in this paper focus on the heterogeneous setting with non-identical value distributions, where the either-or approximation approach becomes necessary.

\subsection{Impossibility results}
Here we present some impossibility results for the revenue of bid-reporting mechanisms and distribution-reporting mechanisms.
\begin{lemma}\label{lemma:impossible-prior-independent}
    For $n\geq 2$ buyers and any $\epsilon\in(0,1)$, for any IC bid-reporting prior-independent mechanism $M$, there exists an instance $\bm{F}\in\calA^{R}$ of identical and regular distributions, such that $\REV(M;\bm{F})<\min\{\epsilon \cdot \WEL,(1+\epsilon)\cdot \SPA\}$.
\end{lemma}
\begin{proof}
Let $H=e^{\frac{n}{\epsilon}}$.
Consider the instance $\bm{F}$ where $F_1,\cdots,F_n$ are identically the truncated equal-revenue distribution defined by quantile function $v(q)=\frac1{\max\{q,\frac1{H}\}}$. We have $\WEL(\bm{F})\geq\E_{v_1\sim F_1}[v_1]=1+\ln H$, and $\SPA(\bm{F})=\Myer(\bm{F})=H(1-(1-\frac1{H})^n)\leq H\cdot n\cdot\frac1{H}=n$ by the property of truncated equal-revenue distribution. It follows that $\Myer(\bm{F})\leq \min\{\frac{n}{1+\ln H}\WEL(\bm{F}),\SPA(\bm{F})\}<\min\{\epsilon\cdot\WEL(\bm{F}),(1+\epsilon)\cdot\SPA(\bm{F})\}$.  It remains to recall that the revenue of any bid-reporting mechanism $M$ is upper bounded by $\Myer(\bm{F})$, which completes the proof.\qed
\end{proof}

To prove the upper bounds on the revenue of distribution-reporting mechanisms, we frequently utilize the following technical observation: When we consider a distribution family that can be viewed as a base distribution multiplied by different factors, any ex-ante IC distribution-reporting mechanism behaves like an IC bid-reporting mechanism, because reporting a distribution in this family is essentially reporting the multiplicative factor. Formally, we have the following lemma:
\begin{lemma}\label{lemma:reduce-to-bid-report}
Let $M=(x,p)$ be any ex-ante IC distribution-reporting mechanism. Consider a base distribution $\hat{F}$ represented by quantile function $\hat{v}(\cdot)$, and a set of multiplicative factors $A\subseteq \bbR_{+}$. For each $a\in A$, let $F^{a}$ denote the distribution obtained by scaling $\hat{F}$ with multiplicative factor $a$, represented by the quantile function $v^a(q)=a\cdot\hat{v}(q)$. 


Assume that $\{F^a:a\in A\}\subseteq \calF_i$ holds for all $i\in[n]$. Then the distribution-reporting mechanism $M$ induces a bid-reporting mechanism $(\tilde{x},\tilde{p})$ on bid space $A$, with allocation and payment rules $\tilde{x}_i:A^n\to\bbR$ and $\tilde{p}_i:A^n\to\bbR$ defined as follows:
\begin{align*}
\tilde{x}_i(a_1,\cdots,a_n)&:=\E_{q_1,\cdots,q_n\sim U[0,1]}[\hat{v}(q_i)x_i(v^{a_1}(q_1),\cdots,v^{a_n}(q_n);F^{a_1},\cdots,F^{a_n})],\\
\tilde{p}_i(a_1,\cdots,a_n)&:=p_i(F^{a_1},\cdots,F^{a_n}),
\end{align*}

For $(\tilde{x},\tilde{p})$, the following incentive compatibility condition holds: For any $\bm{a}_{-i}\in A^{n-1}$, $a_i,a_i'\in A$, it holds that 
$$a_i\cdot\tilde{x}_i(a_i,\bm{a}_{-i})-\tilde{p}_i(a_i,\bm{a}_{-i})\geq a_i\cdot\tilde{x}_i(a_i',\bm{a}_{-i})-\tilde{p}_i(a_i',\bm{a}_{-i}).$$

Moreover, assuming that $M$ is feasible in the single-item setting, then the following feasibility condition for $\tilde{x}$ holds:
For any $a_1,\cdots,a_n\in A$, for each $i\in[n]$, $\tilde{x}_i(a_1,\cdots,a_n)\leq \int_0^1\hat{v}(q_i)dq_i$.
\end{lemma}

\begin{lemma}\label{lemma:impossible-WELandSPA}
    For $n\geq 2$ buyers and $\epsilon>0$, for any valid distribution-reporting mechanism $M$, the following two statements hold:
    \begin{enumerate}
        \item There is an instance $\bm{F}\in\calA^{R}$ such that $\REV(M;\bm{F})<\epsilon\cdot\Myer(\bm{F})\leq\epsilon\cdot\WEL(\bm{F})$.
        \item There is an instance $\bm{F}\in\calA^{R}$ such that $\REV(M;\bm{F})<(1+\epsilon)\cdot\SPA(\bm{F})$.
    \end{enumerate}
    Specifically, the instance $\bm{F}$ can only consist of degenerate distributions, i.e., deterministic values.
\end{lemma}

\section{Main Results: Single Item Setting}\label{sec:main-result}

\subsection{Threshold-Augmented Mechanisms}
In this subsection, we propose the general structure of threshold-augmented mechanisms, which guarantees ex-ante incentive compatibility.

\begin{definition}\label{def:threshold-augmented-mechanism}
A threshold-augmented mechanism $TAM_{M^S,\bm{\tau}}$ is specified by a bid-reporting  mechanism $M^S$ and $n$ threshold functions $\tau_1(\cdot),\cdots,\tau_n(\cdot)$, where
\begin{itemize}
    \item The bid-reporting mechanism is denoted by $M^S=(x^S,p^S)$, where $x_i^S:\bbR_+^n\to [0,1]^n$ and $p_i^S:\bbR_+^n\to \bbR_+^n$ are the allocation and payment rules for buyer $i$.
    \item Each $\tau_i:\prod_{j\neq i}\calF_j\to\mathbb{R}$ calculates a threshold price $\tau_i(\bm{B}_{-i})$ for buyer $i$ based on the distribution of other buyers.
\end{itemize}

Denote $u^S_i=U_i^{M^S}(\bm{B};B_i)=\E_{\bm{v}\sim \bm{B}}[v_i\cdot x^S_i(v_1,\cdots,v_n)-p^S_i(v_1,\cdots,v_n)]$, which is the expected utility  obtained by each buyer $i$ under $M^S$ supposing that the true value distributions are $\bm{B}$. The allocation and payment rules of $TAM_{\bm{\tau}}$ are defined as
$$x_i(\bm{v};\bm{B})=\begin{cases} x^S_i(\bm{v}),&\text{if }u^S_i\geq \tau_i(\bm{B}_{-i}),\\
0,&\text{if }u^S_i< \tau_i(\bm{B}_{-i}).    
\end{cases}$$
and $$p_i(\bm{B})=\begin{cases} \E_{\bm{v}\sim \bm{B}}[p^S_i(\bm{v})]+\tau_i(\bm{B}_{-i}),&\text{if }u^S_i\geq \tau_i(\bm{B}_{-i}),\\
0,&\text{if }u^S_i< \tau_i(\bm{B}_{-i}).    
\end{cases}$$
\end{definition}

In a threshold-augmented mechanism, the threshold for each buyer can be thought of as an ex-ante entry fee that must be paid before entering an SPA with all other buyers. If a buyer's expected utility obtained in the SPA is higher than this threshold, she will pay the threshold to enter. Otherwise, she will opt out, receiving no allocation or payment, but her bid is still simulated in the SPA in competition with other buyers. A key property of threshold-augmented mechanisms is that they guarantee ex-ante incentive compatibility as long as the underlying bid-reporting mechanism $M^S$ is IC in the standard setting.

\begin{theorem}\label{thm:TAM is ex-ante IC}
    If $M^S$ is an IC bid-reporting mechanism, then for any threshold functions $\bm{\tau}$, the threshold-augmented mechanism $TAM_{M^S,\bm{\tau}}$ is ex-ante IC and ex-ante IR.
\end{theorem}


In this section, we adopt SPA as the bid-reporting mechanism in the construction of threshold-augmented mechanisms. 
In the remainder of this section, we assume $M^S=\SPA$ and omit it from the notation, i.e., we abbreviate $TAM_{\SPA,\bm{\tau}}$ as $TAM_{\bm{\tau}}$.
We write $x_i^{\SPA}:\bbR^n\to\{0,1\}$ and $p_i^{\SPA}:\bbR^n\to\bbR$ as the allocation and payment rules of SPA.\footnote{We assume an arbitrary but fixed tie-breaking rule for SPA.}

As a simple example of the threshold-augmented mechanism, we present the following result showing that when all buyers follow an identical but unknown value distribution, the seller can fully extract the optimal social welfare as her revenue by a simple threshold-augmented mechanism. Note that this setting is different from our main model, where the value distribution of different buyers are independent. 

\begin{theorem}\label{thm:iid-REV=WEL}
Assume $n\geq 2$, and all buyers have identical distributions $F_1=\cdots=F_n=F\in\calA$. Consider the threshold-augmented mechanism $TAM_{\bm{\tau}}$, where the threshold function for each buyer $i\in[n]$ is specified as $\tau_i(\bm{B}_{-i})=U^{\SPA}_i(B_{i'},\bm{B}_{-i})$, where $i'$ is arbitrarily chosen in $[n]\setminus\{i\}$. Then, for any distribution $F\in\calA$, $TAM_{\bm{\tau}}$ extracts the full optimal welfare as revenue. Formally, it holds for all $F\in\calA$ that
$$\REV(TAM_{\bm{\tau}},F,\cdots,F)=\mathrm{WEL}(F,\cdots,F).$$
\end{theorem}
\begin{proof}
Given any $F\in\calA$, denote the instance with identical distributions by $\bm{F}=(F,\cdots,F)$.
Since $TAM_{\bm{\tau}}$ is ex-ante IC, all buyers will truthfully report $B_i=F_i$. For each buyer $i\in[n]$, we have $B_{i'}=F=F_i$ regardless of the choice of $i'$, and therefore $\tau_i(\bm{B}_{-i})=U^{\SPA}_i(B_{i'},\bm{B}_{-i})=U^{\SPA}_i(\bm{F})=\E_{\bm{v}\sim \bm{F}}[v_ix_i^{\SPA}(\bm{v})-p_i^{\SPA}(\bm{v})]$.

Denote $TAM_{\bm{\tau}}=(x,p)$. By the definition of threshold-augmented mechanism, $p_i(\bm{B})=\E_{\bm{b}\sim \bm{B}}[p_i^{\SPA}(\bm{b})]+\tau_i(\bm{B}_{-i})=\E_{\bm{v}\sim \bm{F}}[v_ix_i^{\SPA}(\bm{v})]$. Summing over all buyers, we have  
$$\sum_{i\in[n]} p_i(\bm{B})=\E_{\bm{v}\sim \bm{F}}[\sum_{i\in[n]}v_ix_i^{\SPA}(\bm{v})]=\E_{\bm{v}\sim \bm{F}}[\max_{i\in[n]}v_i]=\mathrm{WEL}(\bm{F}).$$
That is, $\REV(TAM_{\bm{\tau}},\bm{F})=\mathrm{WEL}(\bm{F})$.
\qed
\end{proof}
Remarkably, this result demonstrates the advantage of a distribution-reporting mechanism over a bid-reporting mechanism.
By the proof of \Cref{lemma:impossible-WELandSPA}, there exists an instance of identical regular distributions under which no bid-reporting mechanism can achieve $\epsilon\cdot\WEL$ revenue. In contrast, \Cref{thm:iid-REV=WEL} shows that a distribution-reporting mechanism can achieve full welfare extraction under identical distributions.
This highlights how distribution-reporting improves the revenue guarantees in the ex-ante incentive model.

\subsection{Lower bound}
Next, we present our main result for multiple buyers with non-identical distributions. To achieve the either-or approximation, we employ a randomization over threshold-augmented mechanisms to improve the worst-case revenue guarantees. We introduce the Peer-Max mechanism, which sets the threshold for each buyer as the expected maximum value of all other buyers, scaled with a randomized multiplicative factor.
\begin{definition}\label{def:peer-max}
A Peer-Max mechanism $PM_{D_\alpha}$ specified by a distribution $D_{\alpha}$ for parameter $\alpha\in[0,+\infty)$ is defined as the mechanism that draws $\alpha\sim D_\alpha$, and executes the threshold-augmented mechanism $TAM_{\bm{\tau}^{\alpha}}$ with threshold function $\tau_i^{\alpha}(\bm{B}_{-i}):=\alpha\cdot \E_{\bm{v}_{-i}\sim\bm{B}_{-i}}[\max_{j\neq i}v_j]$.

We use $PM_{D_\alpha}(F_1,\cdots,F_n)$ to denote the expected revenue of $PM_{D_\alpha}$ under an instance $(F_1,\cdots,F_n)$.
\end{definition}
One can immediately see that a Peer-Max mechanism is ex-ante IC and ex-ante IR, since it is a randomization over threshold-augmented mechanisms.

\begin{theorem}{\label{theorem:lowerbound}}
Consider $n\geq 2$ buyers, with distribution classes $\mathcal{F}_1=\cdots=\mathcal{F}_n=\calA$ consisting of all distributions.
    
For any integer $K\geq 1$, there exists a distribution $D_\alpha$ for the parameter $\alpha$ of Peer-Max mechanism, such that for all instances $\bm{F}=(F_1,...,F_n)\in\calA^n$, 
$$PM_{D_\alpha}(F_1,\cdots,F_n)\geq \min\left\{\frac{1}{24(K+\log_2 n)}\cdot\mathrm{WEL}(F_1,\cdots,F_n),2^K\cdot\mathrm{SPA}(F_1,\cdots,F_n)\right\}.$$
\end{theorem}

\begin{proofsketch}
We take a parameter $L=\Theta(\log n)$, and construct the distribution $D_{\alpha}$ of $\alpha$ for the Peer-Max mechanism as follows: with probability $\frac12$, let $\alpha=2^{K+1}$, and with the remaining $\frac12$ probability, let $\alpha$ be uniformly randomly drawn from the set $\{0,2^{-L},2^{-L+1},\cdots,\frac12,1,2,4,\cdots,2^K\}$.

Fix any instance $\bm{F}$. For convenience, we omit $\bm{F}$ from the notations. For each buyer $i\in[n]$, we define the following quantities
\begin{align*}
w_i=\E_{\bm{v}\sim \bm F}[v_i\cdot \I[i=\arg\max_{i'\in[n]}v_{i'}]],\ 
s_i=\E_{\bm{v}\sim \bm F}[\max_{j\neq i}v_{j}\cdot \I[i=\arg\max_{i'\in[n]}v_{i'}]],\ 
r_i=\E_{\bm{v}_{-i}\sim\bm{F}_{-i}}[\max_{j\neq i}v_j].
\end{align*}
Here $w_i$ and $s_i$ are buyer $i$'s expected contribution to welfare and expected payment under SPA respectively, and $r_i$ is the expected maximum value of all buyers except $i$. Note that $\alpha\cdot r_i$ is used as the threshold for buyer $i$ in $TAM_{\bm{\tau}^{\alpha}}$, by the definition of $PM_{D_{\alpha}}$. When $\alpha\cdot r_i\leq w_i-s_i$, $TAM_{\bm{\tau}^{\alpha}}$ can extract from buyer $i$ an extra payment of $\alpha\cdot r_i$ compared with SPA. However, when $\alpha\cdot r_i> w_i-s_i$, the buyer $i$ is excluded in $TAM_{\bm{\tau}^{\alpha}}$, causing zero payment.
An important observation is that $r_i=s_i+\sum_{j\neq i}w_i$ for every buyer $i$, that is, the maximum value among all other buyers either becomes the payment of buyer $i$ when buyer $i$ wins in SPA, or becomes the winner's contribution to welfare when buyer $i$ loses in SPA.

We can decompose the benchmarks as $\WEL=\sum_{i\in[n]}w_i$ and $\SPA=\sum_{i\in[n]}s_i$. 
Intuitively, to approximate $\WEL$, we aim to extract a large fraction of $w_i$ as payment from each buyer $i$. This is achieved if $s_i+\alpha\cdot r_i$ is close to but smaller than $w_i$. If $s_i$ is already near $w_i$, a good revenue is extracted when $\alpha=0$. Otherwise, when $s_i$ is small compared to $w_i$, since $r_i=s_i+\sum_{j\neq i}w_i$, we intend to select $\alpha$ so that $\alpha\cdot \sum_{j\neq i}w_i$ approximates $w_i$ from below. However, the relative ratio between $w_i$ and $\sum_{j\neq i}w_i$ may vary: In one extreme, $w_i$ is small compared to $\sum_{j\neq i}w_i$, and we must account for the cases where $w_i=\Omega(\frac1n)\sum_{j\neq i}w_i$, which require $\alpha=\Theta(\frac1n)$ in the smallest; In the other extreme, there is a single buyer $i$ whose $w_i$ is large compared with $\sum_{j\neq i}w_i$, we must account for the cases where $w_i=O(2^K)\sum_{j\neq i}w_i$ with $\alpha=\Theta(2^K)$ in the largest, above which we turn to extract a revenue of $2^K\cdot \SPA$. This motivates us to randomly spread the value of $\alpha$ exponentially over the interval from $\Theta(\frac1n)$ to $\Theta(K)$. Based on this construction, we prove the lower bound through detailed discussions.
\end{proofsketch}

    
    

Remarkably, this result separates distribution-reporting mechanisms from prior-independent mechanisms in revenue guarantees under non-identical distribution settings.
As shown in \Cref{lemma:impossible-prior-independent}, for any $n\geq 2$ and $\epsilon>0$, no DSIC prior-independent mechanism can guarantee a revenue of $\min\{\epsilon \cdot \WEL,(1+\epsilon)\cdot \SPA\}$ even under regular and identical distributions.
In contrast, our positive result shows that distribution-reporting mechanisms can achieve the either-or approximation for general non-identical distributions with positive constant factors given any fixed $n\geq 2$, despite the logarithmic dependence on $n$.

\subsection{Upper bound}
To complement our positive result, we prove a matching upper bound, showing that the either-or approximation result is tight up to constant factors.
\begin{theorem}\label{thm:upperbound}
 Assume $\mathcal{F}_1=\cdots=\mathcal{F}_n=\calA$. For any $n\geq 2$ and $K\geq6$, for any valid distribution-reporting mechanism $M$, 
there exists some instance $\bm{F}=(F_1,\cdots,F_n)\in\bm{\calF}$ such that
    $$\mathrm{REV}(M,F_1,\cdots, F_n)\leq \min\left\{\frac{64}{K+\log_2 n}\cdot\mathrm{WEL}(F_1,\cdots, F_n),2^K\cdot\mathrm{SPA}(F_1,\cdots, F_n)\right\}.$$
\end{theorem}
\begin{proofsketch}
We utilize \Cref{lemma:reduce-to-bid-report} to establish the upper bound on the revenue of any distribution-reporting mechanism $M$.

When $K\geq\log_2 n$, we need to prove that $\mathrm{REV}(M,\bm{F})\leq \min\{O(\frac{1}{K})\cdot\mathrm{WEL}(\bm{F}),2^K\cdot\mathrm{SPA}(\bm{F})\}$ for some instance $\bm{F}$. We utilize a result in \citep{GuruganeshMWW24} showing that for any $T\geq3$, no bid-reporting mechanism can guarantee $\min\{\frac{2.5}{\ln T}\WEL,T\cdot\mathrm{SPA}\}$ revenue in degenerate distribution instances. By \Cref{lemma:reduce-to-bid-report}, any distribution mechanism $M$ behaves essentially the same as some bid-reporting mechanism $M'$ on degenerate distribution instances, so the revenue upper bound is still applicable. 

When $K\leq \log_2 n$, we prove the existence of some instance $\bm{F}$ such that $\mathrm{REV}(M,\bm{F})\leq \min\{O(\frac{1}{\log n})\cdot\mathrm{WEL}(\bm{F}),\Theta(1)\cdot\mathrm{SPA}(\bm{F})\}$. We define a distribution $D_I$ of instance $\bm{F}$, where each $F_i$ is independently drawn from a distribution $D_{F}$ over a family of $L\approx\log_2(\sqrt{n})$ distributions $\{F^{1},\cdots,F^{L}\}$. Each $F^j$ is defined by the quantile function $v^j(q)=\frac{1}{2^j\epsilon}\I[q\leq\epsilon]$, and $D_{F}$ assigns $\delta\cdot 2^j$ probability to each $F^j$, where $\delta=\Theta(\frac{1}{\sqrt{n}})$ is the normalizing factor. 

For the left side, the construction of this distribution family allows us to apply \Cref{lemma:reduce-to-bid-report}, showing that the expected revenue of any distribution-reporting mechanism $M$ when $\bm{F}\sim D_I$ is equal to the expected revenue of some bid-reporting mechanism facing a discrete version of equal revenue distribution, which is upper bounded by $\Theta(\sqrt{n})$.

For the right side, note that $\E_{\bm{F}\sim D_I}[\WEL(\bm{F})]=\Theta(\sqrt{n}\log n)$ and $\E_{\bm{F}\sim D_I}[\SPA(\bm{F})]=\Theta(\sqrt{n})$. Through concentration inequalities, we show that $\Pr_{\bm{F}\sim D_I}[\min\{\Theta(\frac{1}{\log n})\cdot\mathrm{WEL}(\bm{F}),\Theta(1)\cdot\mathrm{SPA}(\bm{F})\}\geq \Theta(\sqrt{n})]=\Omega(1)$, which implies $\E_{\bm{F}\sim D_I}[\min\{\Theta(\frac{1}{\log n})\cdot\mathrm{WEL}(\bm{F}),\Theta(1)\cdot\mathrm{SPA}(\bm{F})\}]\geq\Theta(\sqrt{n})$.

In summary, we obtain $\E_{\bm{F}\sim D_I}[\REV(M,\bm{F})]\leq \E_{\bm{F}\sim D_I}[\min\{\Theta(\frac{1}{\log n})\cdot\mathrm{WEL}(\bm{F}),\Theta(1)\cdot\mathrm{SPA}(\bm{F})\}]$, which implies that there is some instance $\bm{F}$ such that
$\REV(M,\bm{F})\leq  \min\{\Theta(\frac{1}{\log n})\cdot\mathrm{WEL}(\bm{F}),\Theta(1)\cdot\mathrm{SPA}(\bm{F})\}$.
\end{proofsketch}

Additionally, this upper bound can be extended to show that our result remain tight up to constant factors even when all distributions are regular.
\begin{theorem}\label{thm:upperbound-regular}
 Assume $\mathcal{F}_1=\cdots=\mathcal{F}_n=\mathcal{A}^{R}$, i.e. the class of regular distributions. For any $n\geq 2$ and any $K\geq8$, for any valid distribution-reporting mechanism $M$, 
there exists some instance $\bm{F}=(F_1,\cdots,F_n)\in\bm{\calF}$ such that
    $$\REV(M,F_1,\cdots, F_n)\leq \min\left\{\frac{1382.4}{K+\log_2 n}\cdot\mathrm{WEL}(F_1,\cdots, F_n),2^K\cdot\mathrm{SPA}(F_1,\cdots, F_n)\right\}.$$
\end{theorem}

\section{Generalization: Multi-Unit Multi-Demand Setting}\label{sec:generalization}
In this section, we extend our main positive result (\Cref{theorem:lowerbound}) from the single-item setting to a more general multi-unit multi-demand setting, achieving an either-or approximation with respect to optimal welfare $\WEL$ and VCG revenue $\VCG$. 

\begin{definition}
    In a multi-unit multi-demand setting, $m$ units of identical items are sold to $n$ buyers, and each buyer $i\in[n]$ has a demand capacity $d_i$ known publicly. For a distribution-reporting mechanism $M=(x,p)$, the allocation rule $x_i(\bm{b};\bm{B})$ now represents the quantity of items allocated to buyer $i$, and the feasibility requires that
\begin{align*}
    \sum_{i\in[n]}x_i(\bm{b};\bm{B})\leq m, \text{\quad and\quad}
    x_i(\bm{b};\bm{B})\in[0,d_i],~\forall i\in[n].
\end{align*}
We denote the space of feasible allocation by $X_{m,\bm{d}}=\{\bm{z}=(z_1,\cdots,z_n):\sum_{i\in[n]}z_i\leq m\land \forall i\in[n],z_i\in[0,d_i]\}$.
\end{definition}

In the multi-unit multi-demand setting, the utility of each buyer remains as given by \Cref{def:ex-ante utility}, i.e.,
\begin{align*}
   U^M_i(\bm{B};F_i)=\E_{\bm{b}_{-i}\sim \bm{B}_{-i},q_i\sim U[0,1]}[x_i(b_i(q_i),\bm{b}_{-i};\bm B)\cdot v_i(q_i)]-\E_{\bm{b}\sim \bm{B}}[p_i(\bm B)],
\end{align*}
where $v_i(\cdot)$ and $b_i(\cdot)$ denote the quantile function of $F_i$ and $B_i$, respectively.

Denote the allocation and payment rules of VCG by $x_i^{\VCG}(\bm{v})=\arg\max_{\bm{z}\in X_{m,\bm{d}}}\sum_{i\in[n]}z_iv_i$ and $p_i^{\VCG}(\bm{v})=\max_{\bm{z}\in X_{m,\bm{d}}}\sum_{i'\in[n]\setminus\{i\}}z_{i'}v_{i'}-\sum_{i'\in[n]\setminus\{i\}}x_{i'}^{\VCG}(\bm{v})v_{i'}$.\footnote{We assume an arbitrary but fixed tie-breaking rule for the allocation rule of VCG.} Given any instance $\bm{F}$, we define the optimal welfare as $\mathrm{WEL}(\bm{F})=\E_{\bm{v}\sim\bm{F}}[\max_{\bm{z}\in X_{m,\bm{d}}}\sum_{i\in[n]}z_iv_i]$ with a slight abuse of notation, and define the VCG revenue as $\mathrm{VCG}(\bm{F})=\E_{\bm{v}\sim\bm{F}}[p_i^{\VCG}(\bm{v})]$.

To extend our main positive result to the multi-unit multi-demand setting, we introduce the Peer-Welfare Mechanism, which generalizes the Peer-Max Mechanism by replacing SPA with VCG as the underlying mechanism. The threshold function for each buyer $i$ is based on the optimal expected welfare of other buyers ignoring the existence of buyer $i$, which is a generalization of the expected maximum value of other buyers adopted in the Peer-Max Mechanism.

\begin{definition}\label{def:peer-max-vcg}
A Peer-Welfare mechanism $PW_{D_\alpha}$ specified by a distribution $D_{\alpha}$ for parameter $\alpha\in[0,+\infty)$ is defined as the mechanism that draws $\alpha\sim D_\alpha$, and executes the threshold-augmented mechanism $TAM_{\VCG,\bm{\tau}^{\alpha}}$ with threshold function $\tau_i^{\alpha}(\bm{B}_{-i}):=\alpha\cdot \E_{\bm{v}_{-i}\sim\bm{B}_{-i}}[\max_{(z_1,\cdots,z_n)\in X_{m,\bm{d}}}\sum_{i'\in[n]\setminus\{i\}}z_{i'}v_{i'}]$.

We use $PW_{D_\alpha}(F_1,\cdots,F_n)$ to denote the expected revenue of $PW_{D_\alpha}$ under an instance $(F_1,\cdots,F_n)$.
\end{definition}

\begin{theorem}{\label{theorem:lowerbound-vcg}}
Consider $n\geq 2$ buyers, with distribution classes $\mathcal{F}_1=\cdots=\mathcal{F}_n=\calA$ consisting of all distributions.
Given any integer $K\geq 1$, there exists a distribution $D_{\alpha}$ of parameter $\alpha$, such that for all instances $\bm{F}=(F_1,\cdots,F_n)\in\calA^n$,
    $$PW_{D_\alpha}(F_1,\cdots,F_n)\geq \min\{\frac{1}{24(K+\log_2 n)}\cdot\mathrm{WEL}(F_1,\cdots,F_n),2^K\cdot \mathrm{VCG}(F_1,\cdots,F_n)\}.$$
\end{theorem}
\begin{proofsketch}
To generalize \Cref{theorem:lowerbound} to the multi-unit multi-demand setting, we redefine the important quantities $w_i$, $s_i$, and $r_i$ in the proof of \Cref{theorem:lowerbound} as follows:
$$
w_i=\E_{\bm{v}\sim \bm F}[v_i\cdot x_i^{\VCG}(\bm{v})],\ 
s_i=\E_{\bm{v}\sim \bm F}[p_i^{\VCG}(\bm{v})],\ 
r_i=\E_{\bm{v}_{-i}\sim\bm{F}_{-i}}[\max_{\bm{z}\in X_{m,\bm{d}}}\sum_{i'\in[n]\setminus\{i\}}z_{i'}v_{i'}].$$
One can see that the relationship between these quantities and the benchmarks are retained. Specifically,  by the payment rule of VCG, we have $r_i=s_i+\sum_{j\neq i}w_i$. Also, we have $\WEL=\sum_{i\in[n]} w_i$ and $\VCG=\sum_{i\in[n]}s_i$ by definition. This allows us to generalize the analysis in the \Cref{theorem:lowerbound} to this setting.
\end{proofsketch}

We remark that this extension actually holds for any single-parameter environment, as our analysis relies solely on the properties of the VCG mechanism, and does not depend on specific feasibility constraints.

\section{Conclusion}

In this paper, we design ex-ante incentive compatible mechanisms in the distribution-reporting model, with a revenue maximization target evaluated by the either-or approximation approach. We establish matching upper and lower bounds up to constant factors. We then extend our approach to the multi-unit multi-demand setting.
Future work could explore the design of distribution-reporting mechanisms in broader settings, such as heterogeneous multi-item settings, possibly utilizing the family of threshold-augmented mechanisms with different underlying bid-reporting mechanisms. Another interesting direction is investigating the statistical and computational aspects in the distribution-reporting model, such as the sample complexity of the distribution report. Finally, it would be interesting to test these mechanisms in real auction environments to see how well they perform in practice and identify any challenges that arise.

%
%
%
%
%

\bibliographystyle{plainnat}
\bibliography{ref}

\appendix
\newpage

\section{Missing Proofs in Section 2}
\subsection{Proof of \Cref{lemma:reduce-to-bid-report}}

\begin{proof}
Under bid-reporting mechanism $(\tilde{x},\tilde{p})$, consider a buyer $i$ with true value $a_i\in A$ and bidding $a'\in A$, while the bid profile of others is fixed as $\bm{a}_{-i}$. Let $\bm{F}^{\bm{a}_{-i}}$ denote the Cartesian product of $F^{a_j}$ for all $j\neq i$. We denote buyer $i$'s utility by $\tilde{u}_i(a_i',\bm{a}_{-i};a_i)$, which equals
\begin{align*}
\tilde{u}_i(a_i',\bm{a}_{-i};a_i):=&a_i\cdot\tilde{x}_i(a_i',\bm{a}_{-i})-\tilde{p}_i(a_i',\bm{a}_{-i})\\
=&a_i\cdot\E_{q_1,\cdots,q_n\sim U[0,1]}[\hat{v}(q_i)x_i(v^{a_1}(q_1),\cdots,v^{a_n}(q_n);F^{a_i'},\bm{F}^{\bm{a}_{-i}})]-p_i(F^{a_i'},\bm{F}^{\bm{a}_{-i}})\\
=&a_i\cdot\E_{q_i\sim U[0,1],\bm{v}_{-i}\sim \bm{F}^{\bm{a}_{-i}}}\left[ \hat{v}(q_i)x_i(v^{a_i'}(q_i),\bm{v}_{-i};F^{a_i'},\bm{F}^{\bm{a}_{-i}})\right]-p_i(F^{a_i'},\bm{F}^{\bm{a}_{-i}})\\
=&\E_{q_i\sim U[0,1],\bm{v}_{-i}\sim \bm{F}^{\bm{a}_{-i}}}\left[ v^{a_i}(q_i)x_i(v^{a_i'}(q_i),\bm{v}_{-i};F^{a_i'},\bm{F}^{\bm{a}_{-i}})\right]-p_i(F^{a_i'},\bm{F}^{\bm{a}_{-i}})\\
=&U_i^{M}(F^{a_i'},\bm{F}^{\bm{a}_{-i}};F^{a_i})
\end{align*}
Since $M$ is ex-ante IC, we have 
$$U_i^{M}(F^{a_i'},\bm{F}^{\bm{a}_{-i}};F^{a_i})\leq U_i^{M}(F^{a_i},\bm{F}^{\bm{a}_{-i}};F^{a_i}).$$
It follows that $\tilde{u}_i(a_i',\bm{a}_{-i};a_i)\leq \tilde{u}_i(a_i,\bm{a}_{-i};a_i)$, i.e., the bid-reporting mechanism $(\tilde{x},\tilde{p})$ is IC.

Assuming $M$ is feasible in the single item setting, we have 
\begin{align*}
    \tilde{x}_i(a_1,\cdots,a_n)&=\E_{q_1,\cdots,q_n\sim U[0,1]}[\hat{v}(q_i)x_i(v^{a_1}(q_1),\cdots,v^{a_n}(q_n);F^{a_i'},\bm{F}^{\bm{a}_{-i}})]
    \leq\E_{q_i\sim U[0,1]}[\hat{v}(q_i)]
    =\int_0^1\hat{v}(q_i)dq_i,
\end{align*}
where the inequality holds because $x_i(v^{a_1}(q_1),\cdots,v^{a_n}(q_n);F^{a_i'},\bm{F}^{\bm{a}_{-i}})\leq 1$ by the feasibility of $M$.
\qed
\end{proof}
\subsection{Proof of \Cref{lemma:impossible-WELandSPA}}
\begin{proof}
Let $\delta_a$ denote the degenerate distribution at $a$, represented by quantile function $v(q)=a$. 
For the first statement, we apply \Cref{lemma:reduce-to-bid-report} on $M$. In the statement of \Cref{lemma:reduce-to-bid-report}, let the base distribution be the degenerate distribution $\delta_{1}$, and let the range of multiplicative factors be $A=\bbR_+$. Then the induced $(\tilde{x},\tilde{p})$ can be viewed as an IC bid-reporting mechanism $M'$ on bid space $\bbR_+$. Notably, the revenue of $M$ on any instance of degenerate distributions $\delta_{a_1},\cdots,\delta_{a_n}$ equals to the revenue of $M'$ on bid profile $a_1,\cdots,a_n$, formally,
\begin{align*}
    \REV(M,\delta_{a_1},\cdots, \delta_{a_n})=\tilde{p}(a_1,\cdots,a_n).
\end{align*}

Consider the following distribution of $a_1,\cdots,a_n$: Let $a_2=\cdots=a_n=0$, and let $a_1$ follow the truncated equal revenue distribution represented by quantile function $v(q)=\frac1{\max\{q,\frac1{H}\}}$, where $H:=e^{\frac{1}{\epsilon}}$. We have $\E_{a_1,\cdots,a_n}[\tilde{p}(a_1,\cdots,a_n)]\leq 1$ by the property of truncated equal revenue distribution. However, $\E_{a_1,\cdots,a_n}[\Myer(\delta_{a_1},\cdots, \delta_{a_n})]=\E_{a_1}[a_1]=1+\ln H=1+\frac1{\epsilon}>\frac1{\epsilon}$. Therefore, we have
$$\E_{a_1,\cdots,a_n}[\REV(M,\delta_{a_1},\cdots, \delta_{a_n})]=\E_{a_1,\cdots,a_n}[\tilde{p}(a_1,\cdots,a_n)]\leq 1<\epsilon\cdot\E_{a_1,\cdots,a_n}[\Myer(\delta_{a_1},\cdots, \delta_{a_n})].$$
It follows that there exists some $a_1,\cdots,a_n$ such that $$\REV(M,\delta_{a_1},\cdots, \delta_{a_n})<\epsilon\cdot \Myer(\delta_{a_1},\cdots, \delta_{a_n})=\epsilon\cdot\WEL(\delta_{a_1},\cdots, \delta_{a_n}).$$

For the second statement, consider the instance where $F_1=F_2=\delta_1$, and $F_3=\cdots=F_n=\delta_0$. Then $\WEL(F_1,\cdots,F_n)=\SPA(F_1,\cdots,F_n)=1$. By ex-ante IR condition we have $\REV(M,F_1,\cdots,F_n)\leq \WEL(F_1,\cdots,F_n)$, and thus, $\REV(M,F_1,\cdots,F_n)\leq \SPA(F_1,\cdots,F_n)<(1+\epsilon)\cdot\SPA(F_1,\cdots,F_n).$
\qed
\end{proof}

\section{Missing Proofs in Section 3}
\subsection{Proof of \Cref{thm:TAM is ex-ante IC}}
\begin{proof}
Let $x,p$ denote the allocation and payment rules of $TAM_{M^S,\bm{\tau}}$.
    We prove that for each $i\in[n]$, given any $\bm B_{-i}$, the ex-ante utility
    $$U_i^{TAM_{M^S,\bm{\tau}}}(B_i,\bm{B}_{-i};F_i) = \E_{\bm b_{-i}\sim \bm B_{-i}, q_i \sim U[0,1]}[v_i(q_i) \cdot x_i(b_1,\cdots,b_i(q_i),\cdots,b_n; B_i,\bm B_{-i})]-p_i(B_i,\bm B_{-i})$$
    is maximized by truthfully reporting $B_i=F_i$, and is non-negative. 
    Here $v_i(\cdot)$ and $b_i(\cdot)$ represent the quantile function of $F_i$ and $B_i$, respectively.

    By \Cref{def:threshold-augmented-mechanism}, given $B_i$ and $\bm{B}_{-i}$, the ex-ante utility of bidder $i$ with true value distribution $F_i$ is
    \begin{align*}
    & U_i^{TAM_{M^S,\bm{\tau}}}(B_i,\bm{B}_{-i};F_i) \\
    =& \begin{cases}
        \E_{\bm{b}_{-i}\sim \bm B_{-i},q_i\sim U[0,1]} [v_i(q_i)\cdot x^S_i(b_i(q_i),\bm{b}_{-i})-p^S_i(b_i(q_i),\bm{b}_{-i})- \tau_i(\bm B_{-i})] , &\text{if }U_i^{M^S}(\bm{B};B_i)\geq \tau_i(\bm{B}_{-i}),\\
        0, & \text{if }U_i^{M^S}(\bm{B};B_i)< \tau_i(\bm{B}_{-i})
    \end{cases} \\
    =& \begin{cases}
        U_i^{M^S}(B_i,\bm{B}_{-i};F_i)- \tau_i(\bm B_{-i}), &\text{if }U_i^{M^S}(\bm{B};B_i)\geq \tau_i(\bm{B}_{-i}),\\
        0, & \text{if }U_i^{M^S}(\bm{B};B_i)< \tau_i(\bm{B}_{-i})
    \end{cases}
    \end{align*}

Therefore, we have 
$$U_i^{TAM_{M^S,\bm{\tau}}}(B_i,\bm{B}_{-i};F_i)\leq\max\{U_i^{M^S}(B_i,\bm{B}_{-i};F_i)-\tau_i(\bm{B}_{-i}),0\}.$$

On the other hand, when buyer $i$ reports the true distribution, the ex-ante utility is
$$U_i^{TAM_{M^S,\bm{\tau}}}(F_i,\bm{B}_{-i};F_i)=\max\{U_i^{M^S}(F_i,\bm{B}_{-i};F_i)-\tau_i(\bm{B}_{-i}),0\}.$$

Since $M^S$ is a bid-reporting IC mechanism, truthful reporting under $M^S$ maximizes the ex-post utility and thus the ex-ante utility. That is,
$$U_i^{M^S}(F_i,\bm{B}_{-i};F_i)\geq U_i^{M^S}(B_i,\bm{B}_{-i};F_i).$$

It follows that
$$\max\{U_i^{M^S}(F_i,\bm{B}_{-i};F_i),0\}\geq\max\{U_i^{M^S}(B_i,\bm{B}_{-i};F_i),0\}.$$

Therefore, we have $U_i^{TAM_{M^S,\bm{\tau}}}(F_i,\bm{B}_{-i};F_i)\geq U_i^{TAM_{M^S,\bm{\tau}}}(B_i,\bm{B}_{-i};F_i)$, i.e., $TAM_{M^S,\bm{\tau}}$ is ex-ante IC.

Since $U_i^{TAM_{M^S,\bm{\tau}}}(F_i,\bm{B}_{-i};F_i)=\max\{U_i^{M^S}(F_i,\bm{B}_{-i};F_i)-\tau_i(\bm{B}_{-i}),0\}\geq 0$, we have that $TAM_{M^S,\bm{\tau}}$ is ex-ante IR.\qed
\end{proof}
\subsection{Proof of \Cref{theorem:lowerbound}}

\begin{proof}
Define $L=\lceil\log_2(4n)\rceil$. We define the probability distribution $D_{\alpha}$ of $\alpha$ as follows: with probability $\frac12$, let $\alpha=2^{K+1}$, and with the remaining $\frac12$ probability, let $\alpha$ be uniformly randomly drawn from $$A_{L,K}:=\{0,2^{-L},2^{-L+1},\cdots,\frac12,1,2,4,\cdots,2^K\}.$$

Fix any instance $\bm{F}=(F_1,\cdots,F_n)\in\calA^n$, we prove the lower bound on the revenue of $PM_{D_\alpha}$ under $\bm{F}$. We omit $\bm{F}$ from the notation of $\WEL(\bm{F})$ and $\SPA(\bm{F})$.

For each buyer $i\in[n]$, define 
\begin{align*}
w_i=\E_{\bm{v}\sim \bm F}[v_i\cdot \I[i=\arg\max_{i'\in[n]}v_{i'}]],\ 
s_i=\E_{\bm{v}\sim \bm F}[\max_{j\neq i}v_{j}\cdot \I[i=\arg\max_{i'\in[n]}v_{i'}]],\ 
r_i=\E_{\bm{v}_{-i}\sim\bm{F}_{-i}}[\max_{j\neq i}v_j].
\end{align*}
Here $w_i$ and $s_i$ are bidder $i$'s expected contribution to welfare and expected payment under SPA respectively, and $r_i$ is the expected maximum value of all buyers except $i$. The indicator $\mathbb{I}[i = \arg\max_{i' \in [n]} v_{i'}]$ selects $i$ as the unique winner when $v_i$ is maximal, using tie-breaking consistent with the allocation rule of SPA.
Recall that in \Cref{def:peer-max} of the Peer-Max mechanism, $\tau^{\alpha}_i(\bm{F}_{-i})=\alpha\cdot r_i$.

By definition, we have $\WEL=\sum_{i\in[n]}w_i$ and $\SPA=\sum_{i\in[n]}s_i$. For each $i\in[n]$, observe that $$r_i=\E_{\bm{v}\sim\bm{F}}[\max_{j\neq i}v_{j}]=\E_{\bm{v}\sim\bm{F}}[\max_{j\neq i}v_{j}\cdot \I[i=\arg\max_{i'\in[n]}v_{i'}]+\sum_{j\neq i}v_j\cdot \I[j=\arg\max_{i'\in[n]}v_{i'}]]=s_i+\sum_{j\neq i}w_j.$$

For each $i\in[n]$, observe that $s_i \leq w_i$. The following inequalities hold:
\begin{align}\label{eq}
    \SPA\leq r_i\leq \WEL\leq r_i+w_i\leq 2\WEL.
\end{align}
The first inequality is because $\SPA=\sum_{j\in [n]} s_j = s_i+\sum_{j\neq i} s_j$ and $r_i=s_i + \sum_{j\neq i} w_j$. The second inequality follows from $\WEL=w_i + \sum_{j\neq i} w_j$. The last two inequalities are obtained by noticing that $\WEL\le r_i+w_i=\WEL+s_i\le 2\WEL$.

Now we are ready to analyze the revenue of $PM_{D_{\alpha}}$.
For convenience, we denote the revenue of $TAM_{\bm{\tau}^{\alpha}}$ under $\bm{F}$ by $\mathrm{REV}(\alpha)$. By definition, 
$$\mathrm{REV}(\alpha)=\sum_{i\in[n]}\I[w_i\geq s_i+\alpha r_i]\cdot(s_i+\alpha r_i),$$
and $PM_{D_{\alpha}}(\bm{F})=\E_{\alpha\sim D_{\alpha}}[\mathrm{REV}(\alpha)]=\frac12\mathrm{REV}(2^{K+1})+\frac1{2(K+L+1)}\sum_{\alpha\in A_{L,K}}\mathrm{REV}(\alpha)$.

To prove the theorem, we discuss in the following two cases: (A) When there is a buyer with a large welfare contribution $w_i$ in second-price auction compared to the total welfare $\WEL$; (B) When the welfare contribution  of all buyers in SPA are relatively small.

\textbf{case (A):} There exists $i\in[n]$ such that $w_i\geq \frac12\WEL$. Let $\alpha^*$ be the maximal $\alpha\in A_{L,K}\cup\{2^{K+1}\}$ such that $w_i\geq s_i+\alpha r_i$. We prove that either $\E_{\alpha\sim D_{\alpha}}[\mathrm{REV}(\alpha)]\ge 2^K \SPA$ holds, or there exists $\alpha \in A_{L,K}$ such that $\mathrm{REV}(\alpha)\ge \frac 14 \WEL$.

\textbf{case (A.1):} If $\alpha^*=2^{K+1}$, then
$$\mathrm{REV}(\alpha^*) \geq \I[w_i\geq s_i+\alpha^* r_i](s_i+\alpha^* r_i) = s_i+\alpha^* r_i \geq 2^{K+1} r_i\geq 2^{K+1} \SPA$$
where we used \cref{eq} to obtain the last inequality.
In this case, it follows that $$\E_{\alpha\sim D_{\alpha}}[\mathrm{REV}(\alpha)] \geq\frac12 \mathrm{REV}(2^{K+1})\geq 2^{K}\SPA.$$ 

\textbf{case (A.2):} If $\alpha^*=0$, then by the definition of $\alpha^*$, we have $w_i<s_i+2^{-L}r_i\leq s_i+\frac14r_i$. It follows that
\begin{align}\label{caseA.2}
\WEL\leq 2w_i<2s_i+\frac12 r_i\leq 2s_i+\frac12 \WEL,
\end{align}
where the first inequality follows from the assumption of case (A) and the last inequality follows from \cref{eq}.

From \cref{caseA.2}, we have $s_i\geq \frac14 \WEL$. Then $\mathrm{REV}(0)=\SPA\geq s_i\geq \frac14 \WEL$.

\textbf{case (A.3):} If $\alpha^*\in[2^{-L},2^{K}]$, then by the definition of $\alpha^*$, $s_i+\alpha^* r_i\leq w_i<s_i+2\alpha^* r_i$. It follows that
$$
\mathrm{REV}(\alpha^*) \ge s_i+\alpha^* r_i\geq \frac 12 (s_i+2\alpha^* r_i) >\frac12 w_i\geq \frac14 \WEL.
$$

\textbf{case (B):} For all $i\in[n]$, it holds that $w_i\leq \frac12 \WEL$. We focus on a subset of buyers $S=\{i\in [n]:w_i\ge \frac 1{2n}\WEL\}$, i.e., those buyers who each contributes at least $\frac1{2n}$ fraction in the total welfare. Then for any $i\notin S$ we have $w_i<\frac1{2n}\WEL$, and it follows that $$\sum_{i\in S}w_i=\WEL-\sum_{i\notin S}w_i\geq \WEL-\sum_{i\notin S}\frac1{2n}\WEL\geq \WEL-n\cdot\frac1{2n}\WEL=\frac12\WEL.$$

For each buyer $i\in S$,  we have $$w_i\leq \WEL-w_i\leq r_i\leq \WEL\leq 2n w_i.$$ 
The first inequality is because $w_i\leq\frac12\WEL$. The second and third inequalities follow from \cref{eq}. The fourth inequality is because $w_i\geq \frac1{2n}\WEL$.

If $s_i\geq \frac12 w_i$, buyer $i$ contributes at least $\frac12 w_i$ payment to $\mathrm{REV}(0)$. Otherwise, we have $w_i-s_i\in [\frac 12 w_i, w_i]$, and since $r_i \in [w_i, 2nw_i]$, there exists $\alpha^*=2^{-j}$ for some $j\in\{1,\cdots,L\}$ such that $\alpha^* r_i\leq w_i-s_i< 2\alpha^* r_i$. Since $s_i+\alpha^* r_i=\frac12(2s_i+2\alpha^* r_i)>\frac12(s_i+w_i)\geq\frac12 w_i$, buyer $i$ contributes at least $s_i+\alpha^* r_i\geq \frac12 w_i$ payment to $\mathrm{REV}(\alpha^*)$.

In summary, every buyer $i\in S$ contributes at least $\frac 12 w_i$ to $\sum_{\alpha\in {A_{L,K}}}\mathrm{REV}(\alpha)$. So we have
\begin{align*}
\sum_{\alpha\in A_{L,K}}\mathrm{REV}(\alpha)
\geq \frac12\sum_{i\in S}w_i 
\geq \frac14 \WEL.
\end{align*}
It follows that $\E_{\alpha\sim D_{\alpha}}[\mathrm{REV}(\alpha)]\geq\frac{1}{2(K+L+1)}\sum_{\alpha\in A_{L,K}}\mathrm{REV}(\alpha)\geq \frac {1}{8(K+L+1)} \WEL$.

Combining cases (A) and (B) together, we obtain that
$$PM_{D_{\alpha}}(F_1,\cdots,F_n)\geq \min\left\{\frac1{8(K+L+1)}\WEL(F_1,\cdots,F_n),2^K\SPA(F_1,\cdots,F_n)\right\}.$$

Since $L=\lceil \log_2(4n)\rceil\leq \log_2(n)+3$, and $K+\log_2(n)\geq 2$, we have $K+L+1\leq K+\log_2(n)+4\leq 3(K+\log_2(n))$. Therefore, it holds that 
\begin{equation*}
    PM_{D_{\alpha}}(F_1,\cdots,F_n)\geq \min\left\{\frac1{24(K+\log_2(n))}\WEL(F_1,\cdots,F_n),2^K\SPA(F_1,\cdots,F_n)\right\}. \tag*{\qed}
\end{equation*}
\end{proof}

\subsection{Proof of \Cref{thm:upperbound}}
\begin{proof}
The impossibility when $\log_2 n\leq K$ follows directly from Theorem 4.2 in \citep{GuruganeshMWW24}, which states that, for any $n\geq 2$ and $T\geq 3$, for any IC bid-reporting mechanism $M'$, there exist degenerate distributions $\delta_{v_1},\cdots,\delta_{v_n}$ such that
\begin{align*}
    \REV(M',\delta_{v_1},\cdots, \delta_{v_n})\leq \min\left\{\frac{2.5}{\ln T}\Myer(\delta_{v_1},\cdots, \delta_{v_n}),T\cdot\mathrm{SPA}(\delta_{v_1},\cdots, \delta_{v_n})\right\}.
\end{align*}
Here $\delta_{v_i}$ denotes the degenerate distribution at $v_i\in\bbR$, represented by quantile function $v(q)=v_i$.
Note that for degenerate distributions, $\Myer(\delta_{v_1},\cdots, \delta_{v_n})=\WEL(\delta_{v_1},\cdots, \delta_{v_n})=\max_{i\in[n]}v_i$.

To adapt this impossibility result for the distribution-reporting mechanism $M$, we apply \Cref{lemma:reduce-to-bid-report} to show that $M$ behaves like a bid-reporting mechanism on degenerate distributions. In the statement of \Cref{lemma:reduce-to-bid-report}, let the base distribution be the degenerate distribution $\delta_{1}$, represented by quantile function $\hat{v}(q)=1$, and let the range of multiplicative factors be $A=\bbR_+$. Then the corresponding $(\tilde{x},\tilde{p})$ can be viewed as an IC bid-reporting mechanism $M'$, which implies that
\begin{align*}
    \REV(M,\delta_{v_1},\cdots, \delta_{v_n})=\tilde{p}(v_1,\cdots,v_n)\leq \min\left\{\frac{2.5}{\ln T}\WEL(\delta_{v_1},\cdots, \delta_{v_n}),T\cdot\mathrm{SPA}(\delta_{v_1},\cdots, \delta_{v_n})\right\}
\end{align*}

This implies that for any $K\geq\log_23$, if $\log_2 n\leq K$, then there exist $F_1,\cdots,F_n$ such that
\begin{align*}
    \REV(M,F_1,\cdots, F_n)\leq& \min\{\frac{\frac{5}{2}}{K\cdot \ln 2}\mathrm{WEL}(F_1,\cdots, F_n),2^K\mathrm{SPA}(F_1,\cdots, F_n)\}\\
    \leq& \min\left\{\frac{\frac{5}{\ln 2}}{K+\log_2 n}\mathrm{WEL}(F_1,\cdots, F_n),2^K\mathrm{SPA}(F_1,\cdots, F_n)\right\}.
\end{align*}

Now we consider the case where $K\leq \log_2 n$, and then $K+\log n=\Theta(\log n)$. We prove that there is some instance $\bm{F}$ such that
$$\mathrm{REV}(M,\bm{F})\leq \min\left\{\Theta\left(\frac{1}{\log n}\right)\mathrm{WEL}(\bm{F}),\Theta(1)\mathrm{SPA}(\bm{F})\right\}.$$

The following proof divides into three steps:
\begin{description}
    \item[(A)]  Construct a distribution $D_I$ of the instance $\bm{F}=(F_1,\cdots, F_n)$;
    \item[(B)] Prove that $\E_{\bm{F}\sim D_I}[\min\{\frac1{\Theta(\log n)}\mathrm{WEL}(\bm{F}),\mathrm{SPA}(\bm{F})\}] \geq \Omega(n\delta)$ where $\delta = O(\frac{1}{\sqrt n})$;
    \item[(C)] Prove that $\E_{\bm{F}\sim D_I}[\REV(M,\bm{F})]\le O(n\delta)$.
\end{description}

\textbf{Step (A).} Let $L=\lfloor\frac12\log_2 n\rfloor-1$. By assumption $\log_2 n\geq K\geq 6$, so $L\geq 2$, and $\log_2 n\leq 2(L+2)\leq 4L$. 

Let $\epsilon = \frac 1{4n}$. 
Define a distribution $D_{F}$ over a family of $L$ distributions $\{F^{1},\cdots,F^{L}\}$, where each $F^j$ is defined by the quantile function $v^j(q)=\frac{1}{2^j\epsilon}\I[q\leq\epsilon]$. Let $D_{F}$ assign $\delta\cdot 2^j$ probability to each $F^j$, where $\delta:=\frac{1}{\sum_{j=1}^L 2^j}$ is the normalizing factor. Note that $\delta=\frac1{2^{L+1}-2}\geq \frac1{2^{\lfloor\frac12\log_2 n\rfloor}-2}\geq\frac1{\sqrt{n}}$.

Construct the distribution $D_I$ of the instance $\bm{F}=(F_1,\cdots, F_n)$ such that each $F_i$ is independently drawn from distribution $D_{F}$.

\textbf{Step (B).} 
When $F_i\sim D_F$, we have $\E_{F_i\sim D_F}[\E_{v_i\sim F_i}[v_i]]=\sum_{j=1}^{L}\delta\cdot 2^j(\frac{1}{2^j\epsilon}\cdot  \epsilon)=\delta L$. Also, note that regardless of which distribution in the family is drawn as $F_i$, it always holds that $\Pr_{v_i\sim F_i}[v_i>0]=\epsilon$.


When $\bm{F}=(F_1,\cdots,F_n)\sim D_I$, we define the following event
$$E=\left\{\sum_{i\in[n]}\E_{v_i\sim F_i}[v_i]\geq\frac23 n\delta L\right\}.$$
By \Cref{lemma:lemma1} presented later, event $E$ happens with at least $\frac 12$ probability. 
Conditioning on $E$ happens, we have 
\begin{align*}
    \mathrm{WEL}(\bm{F})&\geq \sum_{i\in[n]}\E_{\bm{v}\sim\bm{F}}[v_i\cdot \I[\max_{i'\neq i} v_{i'}=0]]\\
    &\geq \sum_{i\in[n]}(1-(n-1)\epsilon)\E_{v_i\sim F_i}[v_i]\\
    &\geq \frac{3}{4}\sum_{i\in[n]}\E_{v_i\sim F_i}[v_i]\\
    &\geq\frac12 n\delta L.
\end{align*}
The first inequality is by only considering the welfare contribution when at most one value among $v_1,\cdots,v_n$ is positive. The second inequality is because the independence of $v_1,\cdots,v_n$ and $\Pr[\max_{i'\neq i} v_{i'}=0]=(1-\epsilon)^{n-1}\geq 1-(n-1)\epsilon$. The third inequality is by $\epsilon=\frac{1}{4n}$, and the last inequality is because $\sum_{i\in[n]}\E_{v_i\sim F_i}[v_i]\geq\frac23 n\delta L$ assuming event $E$ happens.

Next, for any instance $\bm{F}$, we have
\begin{align*}
    \mathrm{SPA}(\bm{F})\geq&\binom{n}{2}\epsilon^2(1-\epsilon)^{n-2}\frac{1}{2^{L}\epsilon}\\
    \geq& \frac12 n(n-1)\epsilon (1-n\epsilon)\frac{1}{2^L} \\
    = & \frac{3}{32}\frac{n-1}{2^L} 
\end{align*}
The first inequality follows by considering all events $E_{i,j}=\{v_i>0\land v_j>0\land\forall k\in[n]\setminus\{i,j\},v_k=0\}$ with $i,j\in [n]$ and $i<j$. Then $\mathrm {Pr} [E_{i,j}]=\epsilon^2(1-\epsilon)^{n-2}$ and no two events can happen at the same time. When $E_{i,j}$ happens, the revenue of SPA is $\min\{v_i,v_j\}\geq \frac{1}{2^L\epsilon}$. The second inequality follows from Bernoulli inequality: $(1-\epsilon)^{n-2}\ge 1-(n-2)\epsilon\ge 1-n\epsilon$. The last equality is by $\epsilon=\frac1{4n}$.

When event $E$ happens, since $\WEL(\bm{F})\geq \frac12 n\delta L$ and $\SPA(\bm{F})\geq\frac{3}{32}\frac{n-1}{2^L}$, we have
$$\min\{\frac{3}{16L}\mathrm{WEL}(\bm{F}),\mathrm{SPA}(\bm{F})\}\geq n\delta \min\{\frac{3}{32},\frac{3}{32}\frac{n-1}{2^L\delta n}\}.$$
Since $\frac{n-1}{2^L\delta n}=\frac{2^{L+1}-2}{2^L}\frac{n-1}{n}=(2-\frac{2}{2^L})(1-\frac{1}{n})\geq 1$, it follows that
$$\min\{\frac{3}{16L}\mathrm{WEL}(\bm{F}),\mathrm{SPA}(\bm{F})\}\geq \frac{3}{32}n\delta.$$

As $E$ happens with at least $\frac12$ probability, we have
$$\E_{\bm{F}\sim D_I}\left[\min\left\{\frac{3}{16L}\mathrm{WEL}(\bm{F}),\mathrm{SPA}(\bm{F})\right\}\right]\geq\frac{3}{64}n\delta.$$

\textbf{Step (C).} 
To prove the upper bound on the revenue of $M$, we apply \Cref{lemma:reduce-to-bid-report}.
In the statement of \Cref{lemma:reduce-to-bid-report}, consider the base distribution $\hat{F}$ given by quantile function $\hat{v}(q)=\frac1{\epsilon}\I[q\leq\epsilon]$, and the bid space $A=\{\frac1{2^j}:j=1,\cdots,L\}$. Then $D_F$ corresponds to a distribution $D_A$ over $A$ which assigns $\delta\cdot 2^j$ probability to $\frac1{2^j}$.

By \Cref{lemma:reduce-to-bid-report}, $M$ induces an IC bid-reporting mechanism $(\tilde{x},\tilde{p})$. Morever, for each $i\in[n]$, $\tilde{x}_i$ satisfies the feasibility constraint that $\tilde{x}_i(a_1,\cdots,a_n)\leq\int_0^1\hat{v}(q_i)dq_i=1$. 

When $a_1,\cdots,a_n$ are independently drawn from $D_A$, for each buyer $i\in[n]$, the expected revenue of $(\tilde{x},\tilde{p})$ from $i$ is upper-bounded by the optimal posted-price revenue. That is, 
\begin{align*}
    \E_{a_1,\cdots,a_n\sim D_A}[\tilde{p}_i(a_1,\cdots,a_n)]\leq\max_{j\in[L]}\E_{a_i\sim D_A}[\frac1{2^j}\I[a_i\geq \frac1{2^j}]]=\max_{j\in[L]}\frac1{2^j}\sum_{k=1}^j\delta\cdot 2^k
    \leq 2\delta.
\end{align*}

By the definition of $\tilde{p}$, we have
\begin{align*}
    \E_{\bm{F}\sim D_I}[p_i(\bm{F})]=\E_{a_1,\cdots,a_n\sim D_A}[\tilde{p}_i(a_1,\cdots,a_n)]\leq 2\delta.
\end{align*}
It follows that \begin{align*}
    \E_{\bm{F}\sim D_I}[\REV(M,\bm{F})]=\E_{\bm{F}\sim D_I}[\sum_{i\in[n]}p_i(\bm{F})]\leq 2n\delta.
\end{align*}

Combining steps (B) and (C), we have 
$$\E_{\bm{F}\sim D_I}[\REV(M,\bm{F})]\leq 2n\delta\leq \frac{128}{3}\E_{\bm{F}\sim D_I}[\min\{\frac{3}{16L}\mathrm{WEL}(\bm{F}),\mathrm{SPA}(\bm{F})\}].$$

Therefore, there exists some instance $\bm{F}=(F_1,\cdots,F_n)$ such that \begin{align*}
            \mathrm{REV}(M,F_1,\cdots,F_n)&\leq\min\{\frac{8}{L}\mathrm{WEL}(F_1,\cdots,F_n),\frac{128}{3}\mathrm{SPA}(F_1,\cdots,F_n)\}\\
            &\leq \min\{\frac{64}{K+\log_2 n}\mathrm{WEL}(F_1,\cdots,F_n),2^K\mathrm{SPA}(F_1,\cdots,F_n)\}.
        \end{align*}
where the second inequality is because $K+\log_2 n\leq 2\log_2 n\leq 8L$ and $2^K\geq 2^6>\frac{128}{3}$.

This completes the proof of \Cref{thm:upperbound}.\qed
\end{proof}

Lastly, we prove \Cref{lemma:lemma1} used in step (B) of the proof.

\begin{lemma}\label{lemma:lemma1}
    When $\bm{F}\sim D_I$, the event $E=\{\sum_{i\in[n]}\E_{v_i\sim F_i}[v_i]\geq \frac23 n\delta L\}$ happens with at least $\frac12$ probability.
\end{lemma}
\begin{proof}
Define random variable $X_i=\E_{v_i\sim F_i}[v_i]$. We have $\E[X_i]=\sum_{j=1}^L\delta\cdot 2^j \frac{H}{2^j}=\delta L$ and $X_i\in[0,1]$. By Chernoff-Hoeffding's inequality, for any $t\in(0,1)$,
    \begin{align*}
        &\Pr_{\bm{F}\sim D_I}[\sum_{i\in[n]}X_i\leq n\delta L-tn\delta L]\leq e^{\frac{-2(tn\delta L)^2}{n(1-0)^2}}=e^{-2t^2L^2n\delta^2}\leq e^{-2t^2L^2}
    \end{align*}
The last inequality is because $n\delta^2\geq 1$.

Since $L\geq2$, take $t=\frac13$, we have \begin{align*}
    \Pr_{\bm{F}\sim D_I}[\sum_{i\in[n]}\E_{q_i}[v_i(q_i)]\geq \frac23 n\delta LH]\geq 1-e^{-\frac{8}{9}}\geq 0.588\geq \frac12.  \tag*{\qed} 
\end{align*}

\end{proof}

\subsection{Proof of \Cref{thm:upperbound-regular}}

\begin{proof}
When $\log_2 n\leq 3K$, by the proof of \Cref{thm:upperbound} we already know that there exist degenerate distributions (which are regular) $F_1,\cdots,F_n$ such that
\begin{align*}
    \REV(M,F_1,\cdots, F_n)&\leq \min\{\frac{\frac{2.5}{\ln 2}}{K}\mathrm{WEL}(F_1,\cdots, F_n),2^K\mathrm{SPA}(F_1,\cdots, F_n)\}\\
    &\leq \min\{\frac{\frac{10}{\ln 2}}{K+\log_2 n}\mathrm{WEL}(F_1,\cdots, F_n),2^K\mathrm{SPA}(F_1,\cdots, F_n)\}.
\end{align*}

Now we assume $\log_2 n\geq 3K$.
Let $L=\lfloor\frac12\log_2 n\rfloor-1$, we have $L\geq 11$ since $\log_2 n\geq 3K\geq 24$. We prove that there exists $F_1,\cdots,F_n\in\calA^R$ such that $$\REV(M,F_1,\cdots,F_n)\leq \min\{\Theta(\frac{1}{L})\mathrm{WEL}(F_1,\cdots,F_n),\Theta(1)\mathrm{SPA}(F_1,\cdots,F_n)\}.$$

The following proof divides into three steps:
\begin{description}
    \item[(A)]  Construct a distribution $D_I$ of the instance $\bm{F}=(F_1,\cdots, F_n)$;
    \item[(B)] Prove that $\E_{\bm{F}\sim D_I}[\min\{\frac1{\Theta(\log n)}\mathrm{WEL}(\bm{F}),\mathrm{SPA}(\bm{F})\}] \geq \Omega(n\delta\log n)$ where $\delta = O(\frac{1}{\sqrt n})$;
    \item[(C)] Prove that $\E_{\bm{F}\sim D_I}[\REV(M,\bm{F})]\le O(n\delta\log n)$.
\end{description}

\textbf{Step (A).}
Let $\epsilon=\frac1{4n}$. 
Define a distribution $D_{F}$ on a family of $L$ distributions $\{\mathcal{F}^{1},\cdots,\mathcal{F}^{L}\}$, where each $\mathcal{F}^j$ is given by the quantile function $v^j(q)=\frac{1}{2^j}\frac1{\max\{q,\epsilon\}}$. Let $D_{F}$ assign $\delta\cdot 2^j$ probability to each $\mathcal{F}^j$, where $\delta=\frac{1}{\sum_{j=1}^L 2^j}=\frac1{2^{L+1}-2}\geq \frac1{\sqrt{n}}$ is the normalizing factor.

We construct the distribution $D_I$ of the instance $\bm{F}=(F_1,\cdots, F_n)$ such that each $F_i$ is independently drawn from distribution $D_{F}$.
When $\bm{F}\sim D_I$, we define random variable $n_j:=|\{i\in[n]:F_i=F^j\}|$ as the number of buyers who are assigned $F^j$ as value distribution.

\textbf{Step (B).}
Define $L_1=\lceil\frac13 L\rceil$, $L_2=\lfloor\log_2(\frac1{\delta L})\rfloor$. Note that $L_2-L_1=\frac23 L-O(\log\log n)=\Theta(\log n)$.


We define the following two events that both happen with high probability, via concentration inequalities:
\begin{align*}
&E_1=\{\sum_{j=1}^{L}n_j\frac1{2^j}\leq\frac43 n\delta L\},\\
&E_2=\{\sum_{j=1}^{L_1}n_j\frac1{2^j}\geq\frac12 n\delta L_1\}.
\end{align*}

By \Cref{lemma:regular-concentration1} and \Cref{lemma:regular-concentration2} presented later, we have $\Pr_{\bm{F}\sim D_I}[E_1]\geq\frac34$ and $\Pr_{\bm{F}\sim D_I}[E_2]\geq\frac34$. It follows that $\Pr_{\bm{F}\sim D_I}[E_1\cap E_2]\geq1-(1-\Pr_{\bm{F}\sim D_I}[E_1])-(1-\Pr_{\bm{F}\sim D_I}[E_2])\geq \frac12$, i.e., with constant probability both $E_1$, $E_2$ happens.

When events $E_1$ and $E_2$ happens, we derive a lower bound on $\WEL$:
\begin{align*}
    \WEL(\bm{F})=&\E_{\bm{v}\sim\bm{F}}\left[\max_{i\in[n]}v_i\right]\\
    \geq& \sum_{j=L_1}^{L_2}(\frac{1}{2^j\epsilon}-\frac{1}{2^{j+1}\epsilon})\Pr_{\bm{v}\sim\bm{F}}\left[\max_{i\in[n]}v_i\geq \frac{1}{2^j\epsilon}\right]\\
    =& \sum_{j=L_1}^{L_2}\frac{1}{2^{j+1}\epsilon}\left(1-\prod_{k=1}^{j}\left(1-2^{j-k}\epsilon\right)^{n_k}\right)\\
    \geq& \sum_{j=L_1}^{L_2}\frac{1}{2^{j+1}\epsilon}\left(\sum_{k=1}^{j}n_k2^{j-k}\epsilon\right)\left(1-\sum_{k=1}^{j}n_k2^{j-k}\epsilon\right)\\
    =& \sum_{j=L_1}^{L_2}\frac{1}{2}\left(\sum_{k=1}^{j}n_k2^{-k}\right)\left(1-2^j\epsilon\sum_{k=1}^{j}n_k2^{-k}\right)\\
    \geq&\sum_{j=L_1}^{L_2}\frac{1}{4}n\delta L_1(1-2^j\epsilon \frac43 n\delta L)\\
    \geq&(L_2-L_1+1)\frac16n\delta L_1
\end{align*}
The first inequality is because $\E_{\bm{v}\sim\bm{F}}[\max_{i\in[n]}v_i]=\int_0^{+\infty}\Pr[\max_{i\in[n]}v_i\geq t]dt$. The second inequality is because $1-\prod_{i=1}^n (1-t_i)=\sum_{i=1}^nt_i\prod_{j=1}^{i-1}(1-t_j)\geq \sum_{i=1}^nt_i\prod_{j=1}^{n}(1-t_j)\geq \sum_{i=1}^nt_i(1-\sum_{j=1}^{n}t_j)$ holds for any variables $t_1,\cdots,t_n\in[0,1]$. The third inequality is because when $E_1$ and $E_2$ happens, we have $\sum_{k=1}^{j}n_k2^{-k}\geq \sum_{k=1}^{L_1}n_k2^{-k}\geq\frac12n\delta L_1$ and $\sum_{k=1}^{j}n_k2^{-k}\leq \sum_{k=1}^{L}n_k2^{-k}\leq\frac43n\delta L$ for $L_1\leq j\leq L_2$. The last inequality is because $2^j\epsilon \frac43 n\delta L\leq 2^{L_2}\frac1{4n}\frac43 n\delta L\leq \frac13$, recalling $L_2=\lfloor\log_2(\frac1{\delta L})\lfloor$ and $\epsilon=\frac1{4n}$.

To derive a lower bound on $\SPA(\bm{F})$, we define $\varphi_{F_i}(v_i)$ as the Myerson virtual value function of distribution $F_i$. By the property of truncated equal-revenue distributions, we have $\varphi_{F^j}(v_i)=\I[v_i=\frac1{2^j\epsilon}]\frac1{2^j\epsilon}$. Then we have
\begin{align*}
    \SPA(\bm{F})\geq&\E_{\bm{v}\sim\bm{F}}\left[\sum_{i\in[n]}\varphi_{F_i}(v_i)\I[v_i>\max_{i'\neq i}v_{i'}]\right]\\
    \geq&\sum_{j=1}^{L}n_j\epsilon\frac1{2^j\epsilon}\Pr\left[\max_{i\in[n]}v_i<\frac1{2^j\epsilon}\right]\\
    \geq&\sum_{j=1}^{L_2}n_j2^{-j}(1-2^j\epsilon\sum_{k=1}^{j}n_k2^{-k})\\
    \geq&\sum_{j=1}^{L_2}n_j2^{-j}(1-2^j\epsilon\frac{4}{3}n\delta L)\\
    \geq&\frac23\sum_{j=1}^{L_2}n_j2^{-j}\\
    \geq&\frac13n\delta L_1
\end{align*}
The first inequality is by only considering the virtual welfare contribution when $v_i$ is strictly maximal among $v_1,\cdots,v_n$. The second inequality is because $\E_{v_i\sim F^j}[\varphi_{F^j}(v_i)\I[v_i>\max_{i'\neq i}v_{i'}]]=\epsilon\frac1{2^j\epsilon}\Pr[\max_{i'\neq i}v_{i'}<\frac1{2^j\epsilon}]\geq \epsilon\frac1{2^j\epsilon}\Pr[\max_{i'\in[n]}v_{i'}<\frac1{2^j\epsilon}]$ when $F_i=F^j$. The third inequality is because $\Pr[\max_{i\in[n]}v_i<\frac1{2^j\epsilon}]=\prod_{k=1}^j(1-2^{j-k}\epsilon)^{n_k}\geq 1-\sum_{k=1}^jn_k 2^{j-k}\epsilon$. The fourth inequality is because $\sum_{k=1}^{j}n_k2^{-k}\leq \sum_{k=1}^{L}n_k2^{-k}\leq\frac43n\delta L$ when $E_1$ happens. The fifth inequality is by $\epsilon=\frac1{4n}$ and $2^j\leq 2^{L_2}\leq\frac1{\delta L}$. The last inequality is because $\sum_{j=1}^{L_2}n_j2^{-j}\geq \sum_{j=1}^{L_1}n_j2^{-j}\geq \frac12 n\delta L_1$ when $E_2$ happens.

Therefore, when both $E_1$ and $E_2$ happens, we have
$$\min\{\frac2{L_2-L_1+1}\WEL(\bm{F}),\SPA(\bm{F})\}\geq\frac13n\delta L_1.$$

Since $\Pr_{\bm{F}\sim D_I}[E_1\cap E_2]\geq \frac12$, we have
$$\E_{\bm{F}\sim D_I}[\min\{\frac2{L_2-L_1+1}\WEL(\bm{F}),\SPA(\bm{F})\}]\geq\frac16n\delta L_1.$$

\textbf{Step (C).} 
To prove the upper bound on the revenue of $M$, we apply \Cref{lemma:reduce-to-bid-report}.
In the statement of \Cref{lemma:reduce-to-bid-report}, consider the base distribution $\hat{F}$ given by quantile function $\hat{v}(q)=\frac1{\max\{q,\epsilon\}}$, and the bid space $A=\{\frac1{2^j}:j=1,\cdots,L\}$. Then $D_F$ corresponds to a distribution $D_A$ over $A$ which assigns $\delta\cdot 2^j$ probability to $\frac1{2^j}$.

By \Cref{lemma:reduce-to-bid-report}, $M$ induces an IC bid-reporting mechanism $(\tilde{x},\tilde{p})$. Morever, for each $i\in[n]$, $\tilde{x}_i$ satisfies the feasibility constraint that $\tilde{x}_i(a_1,\cdots,a_n)\leq\int_0^1\hat{v}(q_i)dq_i=1+\ln\frac{1}{\epsilon}$. 

When $a_1,\cdots,a_n$ are independently drawn from $D_A$, for each buyer $i\in[n]$, the expected revenue of $(\tilde{x},\tilde{p})$ from $i$ is upper-bounded by the optimal posted-price revenue. That is, 
\begin{align*}
    \E_{a_1,\cdots,a_n\sim D_A}\left[\tilde{p}_i(a_1,\cdots,a_n)\right]
    &\leq\max_{j\in[L]}\E_{a_i\sim D_A}\left[\left(1+\ln\frac{1}{\epsilon}\right)\frac1{2^j}\I[a_i\geq \frac1{2^j}]\right]\\
    &=\left(1+\ln\frac{1}{\epsilon}\right)\max_{j\in[L]}\frac1{2^j}\sum_{k=1}^j\delta\cdot 2^k\\
    &=\left(1+\ln\frac{1}{\epsilon}\right)\max_{j\in[L]}\frac1{2^j}\delta(2^{j+1}-2)\\
    &\leq 2\delta\left(1+\ln\frac{1}{\epsilon}\right).
\end{align*}

By the definition of $\tilde{p}$, we have
\begin{align*}
    \E_{\bm{F}\sim D_I}[p_i(\bm{F})]=\E_{a_1,\cdots,a_n\sim D_A}[\tilde{p}_i(a_1,\cdots,a_n)]\leq 2\delta\left(1+\ln\frac{1}{\epsilon}\right).
\end{align*}
It follows that \begin{align*}
    \E_{\bm{F}\sim D_I}\left[\REV(M,\bm{F})\right]=\E_{\bm{F}\sim D_I}\left[\sum_{i\in[n]}p_i(\bm{F})\right]\leq 2n\delta\left(1+\ln\frac{1}{\epsilon}\right).
\end{align*}

Combining steps (B) and (C), we have
\begin{align*}
\E_{\bm{F}\sim D_I}[\REV(M,\bm{F})]&\leq 2n\delta(1+\ln\frac{1}{\epsilon})\leq\frac{12(1+\ln\frac1{\epsilon})}{L_1}\E_{\bm{F}\sim D_I}[\min\{\frac2{L_2-L_1+1}\WEL(\bm{F}),\SPA(\bm{F})\}].
\end{align*}

Since $L=\lfloor \frac12\log_2 n\rfloor-1\geq 11$, we have $\log_2 n\leq 2(L+2)=2L+4\leq 2.4L$, and $K+\log_2 n\leq \frac43\log_2 n\leq 3.2 L$. We also have $1+\ln \frac1{\epsilon}=1+\ln(2)\log_2(4n)=1+2\ln 2+\ln2\cdot\log_2(n)\leq 1+6\ln 2+2\ln 2\cdot L<5.2+1.4L< 2L$. 
Additionally, we have $L_2=\lfloor\log_2(\frac1{\delta L})\rfloor=\lfloor\log_2(\frac{2^{L+1}-2}{L})\rfloor\geq \log_2(2^{L+1}(1-2^{-L}))-\log_2 L-1= L-\log_2(\frac{L}{1-2^{-L}})\geq\frac23 L$.
Therefore, $L_2-L_1+1\geq \frac23 L-\frac13 L=\frac13 L$.

Therefore, there exists a distribution profile $\bm{F}=(F_1,\cdots,F_n)$ of regular distributions, such that
\begin{align*}
\REV(M,\bm{F})&\leq \frac{12(1+\ln\frac1{\epsilon})}{L_1}\min\left\{\frac2{L_2-L_1+1}\cdot \WEL(\bm{F}),\SPA(\bm{F})\right\}\\
&\leq \frac{24L}{\frac13 L}\min\{\frac2{\frac13 L}\cdot \WEL(\bm{F}),\SPA(\bm{F})\}\\
&\leq\min\left\{\frac{432}{L}\cdot \WEL(\bm{F}),72\cdot \SPA(\bm{F})\right\}\\
&\leq \min\left\{\frac{1382.4}{K+\log_2 n}\cdot \WEL(\bm{F}),2^K\cdot \SPA(\bm{F})\right\}.
 \tag*{\qed}
\end{align*}

Lastly, we prove \Cref{lemma:regular-concentration1} and \Cref{lemma:regular-concentration2} used in step (B) of the proof.

\begin{lemma}\label{lemma:regular-concentration1}
    When $\bm{F}\sim D_I$, the event $E_1=\{\sum_{j=1}^{L}n_j\frac1{2^j}\leq\frac43 n\delta L\}$ happens with at least $\frac34$ probability.
\end{lemma}
\begin{proof}
Define random variable $X_i=\sum_{j=1}^L\I[F_i=F^j]\frac1{2^j}$, then $\sum_{j=1}^{L}n_j\frac1{2^j}=\sum_{i\in[n]}X_i$. We have $\E[X_i]=\sum_{j=1}^L\delta 2^j\cdot\frac1{2^j}=\delta L$ and $X_i\in[0,\frac12]$. By Chernoff-Hoeffding's inequality, for any $t\in(0,1)$,
\begin{align*}
    &\Pr_{\bm{F}\sim D_I}[\sum_{i\in[n]}X_i\geq n\delta L+tn\delta L]\leq e^{\frac{-2(tn\delta L)^2}{n(\frac12-0)^2}}=e^{-8t^2L^2n\delta^2}\leq e^{-8t^2L^2}.
\end{align*}
The last inequality is because $n\delta^2\geq 1$.

Since $L> 3$, take $t=\frac13$, we have 
\begin{align*}
    \Pr_{\bm{F}\sim D_I}[\sum_{i\in[n]}X_i\leq\frac43 n\delta L]\geq 1-e^{-8}\geq 0.999\geq \frac34.
\end{align*}
That is, $\Pr_{\bm{F}\sim D_I}[E_1]\geq\frac34$.
\qed
\end{proof}

\begin{lemma}\label{lemma:regular-concentration2}
    When $\bm{F}\sim D_I$, the event $E_2=\{\sum_{j=1}^{L_1}n_j\frac1{2^j}\geq\frac12 n\delta L_1\}$ happens with at least $\frac34$ probability.
\end{lemma}
\begin{proof}
Define random variable $X_i=\sum_{j=1}^{L_1}\I[F_i=F^j]\frac1{2^j}$, then $\sum_{j=1}^{L_1}n_j\frac1{2^j}=\sum_{i\in[n]}X_i$. We have $\E[X_i]=\sum_{j=1}^L\delta 2^j\cdot\frac1{2^j}=\delta L_1$ and $X_i\in[0,\frac12]$. By Chernoff-Hoeffding's inequality, for any $t\in(0,1)$,
    \begin{align*}
        &\Pr_{\bm{F}\sim D_I}\left[\sum_{i\in[n]}X_i\leq n\delta L_1-tn\delta L_1\right]\leq e^{\frac{-2(tn\delta L_1)^2}{n(\frac12-0)^2}}=e^{-8t^2L_1^2n\delta^2}\leq e^{-8t^2L_1^2}.
    \end{align*}
The last inequality is because $n\delta^2\geq 1$.

Since $L_1\geq\frac13 L\geq 3$, take $t=\frac12$, we have 
\begin{align*}
    \Pr_{\bm{F}\sim D_I}\left[\sum_{i\in[n]}X_i\geq \frac12 n\delta L_1\right]\geq 1-e^{-18}\geq 0.999\geq \frac34.
\end{align*}
That is, $\Pr_{\bm{F}\sim D_I}[E_2]\geq\frac34$.
\qed
\end{proof}
\end{proof}

\section{Missing Proofs in Section 4}
\subsection{Proof of \Cref{theorem:lowerbound-vcg}}
\begin{proof}
Define $L=\lceil\log_2(4n)\rceil$. We define the probability distribution $D_{\alpha}$ of $\alpha$ as follows: with probability $\frac12$, let $\alpha=2^{K+1}$, and with the remaining $\frac12$ probability, let $\alpha$ be uniformly randomly drawn from $$A_{L,K}:=\{0,2^{-L},2^{-L+1},\cdots,\frac12,1,2,4,\cdots,2^K\}.$$

Fix any $\bm{F}=(F_1,\cdots,F_n)\in\calA^n$, we prove the lower bound on the revenue of $PW_{D_\alpha}$ under $\bm{F}$. We omit $\bm{F}$ from the notation of $\WEL(\bm{F})$ and $\VCG(\bm{F})$.

For each buyer $i\in[n]$, define 
\begin{align*}
w_i=\E_{\bm{v}\sim \bm F}[v_i\cdot x_i^{\VCG}(\bm{v})],\ 
s_i=\E_{\bm{v}\sim \bm F}[p_i^{\VCG}(\bm{v})],\ 
r_i=\E_{\bm{v}_{-i}\sim\bm{F}_{-i}}[\max_{\bm{z}\in X_{m,\bm{d}}}\sum_{i'\in[n]\setminus\{i\}}z_{i'}v_{i'}].
\end{align*}
Here $w_i$ and $s_i$ are bidder $i$'s expected contribution to welfare and expected payment under VCG respectively, and $r_i$ is the optimal welfare of all other buyers ignoring the existence of buyer $i$, with $X_{m,\bm{d}}$ denoting the set of feasible allocations. 
Recall that in \Cref{def:peer-max-vcg} of the Peer-Welfare mechanism, $\tau^{\alpha}_i(\bm{F}_{-i})=\alpha\cdot r_i$.

Observe that, by the payment rule of VCG mechanism, it holds for each $i\in[n]$ that
\begin{align*}
    p^{\VCG}_i(\bm{v})=\max_{\bm{z}\in X_{m,\bm{d}}}\sum_{i'\in[n]\setminus\{i\}}z_{i'}v_{i'}-\sum_{i'\in[n]\setminus\{i\}}x_{i'}^{\VCG}(\bm{v})v_{i'}
\end{align*}
Taking expectation over $\bm{v}\sim\bm{F}$, we have $s_i=r_i-\sum_{j\neq i}w_j$, and it follows that
$$r_i=s_i+\sum_{j\neq i}w_j.$$


The following analysis steps are similar to the proof of \Cref{theorem:lowerbound}.
We have $\WEL=\sum_{i\in[n]}w_i$ and $\VCG=\sum_{i\in[n]}s_i$, and  $s_i\leq w_i$ for each $i\in[n]$. It follows that 
\begin{equation}\label{eq:2}
    \VCG\leq r_i\leq \WEL\leq r_i+w_i\leq 2\WEL.
\end{equation}

We denote the revenue of $TAM_{\VCG,\bm{\tau}^{\alpha}}$ under $\bm{F}$ by $\mathrm{REV}(\alpha)$. By definition, we have $$\mathrm{REV}(\alpha)=\sum_{i\in[n]}\I[w_i\geq s_i+\alpha r_i]\cdot(s_i+\alpha r_i),$$
and $PW_{D_{\alpha}}(\bm{F})=\E_{\alpha\sim D_{\alpha}}[\mathrm{REV}(\alpha)]=\frac12\mathrm{REV}(2^{K+1})+\frac1{2(K+L+1)}\sum_{\alpha\in A_{L,K}}\mathrm{REV}(\alpha)$.

To prove the theorem, we discuss in two cases: (A) When there is a buyer whose welfare contribution $w_i$ in VCG contributes a large fraction in the total welfare $\WEL$; (B) When the contributed welfare of all buyers in VCG are relatively small.

\textbf{case (A):} There exists $i\in[n]$ such that $w_i\geq \frac12\WEL$. Let $\alpha^*$ be the maximal $\alpha\in A_{L,K}\cup\{2^{K+1}\}$ such that $w_i\geq s_i+\alpha r_i$. We prove that either $\E_{\alpha\sim D_{\alpha}}[\mathrm{REV}(\alpha)]\ge 2^K \VCG$ holds, or there exists $\alpha \in A_{L,K}$ such that $\mathrm{REV}(\alpha)\ge \frac 14 \WEL$.

\textbf{case (A.1):} If $\alpha^*=2^{K+1}$, then
$$\mathrm{REV}(\alpha^*) \geq \I[w_i\geq s_i+\alpha^* r_i](s_i+\alpha^* r_i) = s_i+\alpha^* r_i \geq 2^{K+1} r_i\geq 2^{K+1} \VCG,$$
where the last inequality is by \cref{eq:2}.
In this case, it follows that $$\E_{\alpha\sim D_{\alpha}}[\mathrm{REV}(\alpha)] \geq\frac12 \mathrm{REV}(2^{K+1})\geq 2^{K}\VCG.$$ 

\textbf{case (A.2):} If $\alpha^*=0$, then by the definition of $\alpha^*$, we have $w_i<s_i+2^{-L}r_i\leq s_i+\frac14r_i$. It follows that
\begin{align*}
\WEL\leq 2w_i<2s_i+\frac12 r_i\leq 2s_i+\frac12 \WEL,
\end{align*}
where the first inequality follows from the assumption of case (A) and the last inequality follows from \cref{eq:2}.
This implies that $s_i\geq \frac14 \WEL$. Then we have $\mathrm{REV}(0)=\VCG\geq s_i\geq \frac14 \WEL$.

\textbf{case (A.3):} If $\alpha^*\in[2^{-L},2^{K}]$, then by the definition of $\alpha^*$, $s_i+\alpha^* r_i\leq w_i<s_i+2\alpha^* r_i$. It follows that
$$
\mathrm{REV}(\alpha^*) \ge s_i+\alpha^* r_i\geq \frac 12 (s_i+2\alpha^* r_i) >\frac12 w_i\geq \frac14 \WEL.
$$

\textbf{case (B):} For all $i\in[n]$, it holds that $w_i\leq \frac12 \WEL$. We focus on a subset of buyers $S=\{i\in [n]:w_i\ge \frac 1{2n}\WEL\}$, i.e., those buyers who each contributes at least $\frac1{2n}$ fraction in the total welfare. Then we have $\sum_{i\in S}w_i=\WEL-\sum_{i\notin S}w_i\geq \WEL-\sum_{i\notin S}\frac1{2n}\WEL\geq \WEL-n\cdot\frac1{2n}\WEL=\frac12\WEL$.

For each buyer $i\in S$,  we have $$w_i\leq \WEL-w_i\leq r_i\leq \WEL\leq 2n w_i.$$ 
The first inequality is because $w_i\leq\frac12\WEL$. The second and third inequalities follow from \cref{eq:2}. The fourth inequality is because $w_i\geq \frac1{2n}\WEL$.

If $s_i\geq \frac12 w_i$, buyer $i$ contributes at least $\frac12 w_i$ payment to $\mathrm{REV}(0)$. Otherwise, we have $w_i-s_i\in [\frac 12 w_i, w_i]$, and since $r_i \in [w_i, 2nw_i]$, there exists $\alpha^*=2^{-j}$ for some $j\in\{1,\cdots,L\}$ such that $\alpha^* r_i\leq w_i-s_i< 2\alpha^* r_i$. Since $s_i+\alpha^* r_i=\frac12(2s_i+2\alpha^* r_i)>\frac12(s_i+w_i)\geq\frac12 w_i$, buyer $i$ contributes at least $s_i+\alpha^* r_i\geq \frac12 w_i$ payment to $\mathrm{REV}(\alpha^*)$.

In summary, every buyer $i\in S$ contributes at least $\frac 12 w_i$ to $\sum_{\alpha\in {A_{L,K}}}\mathrm{REV}(\alpha)$. So we have
\begin{align*}
\sum_{\alpha\in A_{L,K}}\mathrm{REV}(\alpha)
\geq \frac12\sum_{i\in S}w_i 
\geq \frac14 \WEL.
\end{align*}
It follows that $\E_{\alpha\sim D_{\alpha}}[\mathrm{REV}(\alpha)]\geq\frac{1}{2(K+L+1)}\sum_{\alpha\in A_{L,K}}\mathrm{REV}(\alpha)\geq \frac {1}{8(K+L+1)} \WEL$.

Combining cases (A) and (B) together, we obtain that
$$PW_{D_{\alpha}}(F_1,\cdots,F_n)\geq \min\left\{\frac1{8(K+L+1)}\WEL(F_1,\cdots,F_n),2^K\VCG(F_1,\cdots,F_n)\right\}.$$

Since $L=\lceil \log_2(4n)\rceil\leq \log_2(n)+3$, and $K+\log_2(n)\geq 2$, we have $K+L+1\leq K+\log_2(n)+4\leq 3(K+\log_2(n))$. Therefore, it holds that 
\begin{equation*}
    PW_{D_{\alpha}}(F_1,\cdots,F_n)\geq \min\left\{\frac1{24(K+\log_2(n))}\WEL(F_1,\cdots,F_n),2^K\VCG(F_1,\cdots,F_n)\right\}.
\tag*{\qed}
\end{equation*}
\end{proof}

\end{document}